\documentclass[11pt,a4paper]{article}
\usepackage[margin=2.2cm]{geometry}
\usepackage{amsmath, amssymb, amsthm}
\usepackage{booktabs, array}
\usepackage{graphicx}
\usepackage{float}
\usepackage[hidelinks]{hyperref}
\usepackage[T1]{fontenc}
\usepackage{microtype}
\usepackage{enumitem}
\usepackage{tikz}
\usetikzlibrary{calc}

\newtheorem{theorem}{Theorem}
\newtheorem{proposition}[theorem]{Proposition}
\newtheorem{lemma}[theorem]{Lemma}
\newtheorem{corollary}[theorem]{Corollary}
\theoremstyle{definition}

\theoremstyle{remark}
\newtheorem{remark}[theorem]{Remark}

\newcommand{\dd}{\,d}

\newif\ifblind\blindfalse

\title{Closed forms and open obstructions:\\
the sample variance of three observations}
\ifblind
\author{}
\else
\author{%
  Remus Osan\thanks{Gilead Sciences, Foster City, CA 94404, USA.
    \texttt{remus.osan@gilead.com}. Corresponding author.}
  \and Kevin T. Chu\thanks{Velexi Corporation, Burlingame, CA 94010, USA.
    \texttt{kevin@velexi.com}.}
  \and Ron Yu\thanks{Gilead Sciences, Foster City, CA 94404, USA.
    \texttt{ron.yu@gilead.com}.}%
}
\fi
\date{}

\ifblind
\else
\fi
\begin{document}
\maketitle

\begin{abstract}
The exact distribution of the sample variance for $n=3$ is known since Rietz
(1931) for the uniform parent, and Royen (2007, 2008) gave a Fourier series for
any bounded continuous parent.  Neither settles which parents admit a finite
closed form, nor which special functions it forces.  In coordinates aligned with the cube diagonal the variance constraint becomes a
cylinder and the cube cross-section a polygon with $S_3$ symmetry, hexagonal over the central
band of the diagonal and triangular near either corner; the CDF is the
volume of their intersection, the parent density entering as a weight.  A closure hierarchy is
then organized by the minimal function class containing the parent density: polynomial parents on any bounded interval always close in
elementary terms, a theorem for the whole class; in the negative direction a single explicit
parent already suffices, and for a rational one we \emph{prove} the CDF is not elementary, the
obstruction being an irreducible dilogarithmic part.  Beyond those two theorems the hierarchy is
a set of example-specific obstructions rather than a classification: for an algebraic parent the
radial first-kind differential is shown non-elementary on a genus-two curve, and the
exponential row is a conjecture supported by the Bessel structure of its radial integral.  For
the uniform parent we obtain a two-piece formula bifurcating at $Y=1/4$, where the variance
disk first circumscribes the hexagonal cross-section at the cube center.  For the singular
arcsine parent we derive both endpoint laws in closed form and give a six-term approximation accurate to about $10^{-3}$, an accuracy Royen's universal series reaches at about a hundred terms.
\end{abstract}



%=====================================================================
\section{Introduction}\label{sec:intro}
%=====================================================================

\paragraph{Motivation.}
Content uniformity (CU) testing in pharmaceutical manufacturing
assesses whether individual dosage units in a batch have
acceptably low variability.  Regulatory guidelines (USP $\langle
905\rangle$~\cite{usp905}, Ph.~Eur.~2.9.40~\cite{pheur2940}) base acceptance on
the sample variance of small samples ($n = 10$ or $n = 30$) drawn
from the batch~\cite{williamsHauck2002}.  Standard practice assumes a Gaussian
parent~\cite{bergum1990}, but empirical evidence shows that the
true within-batch density of drug content per dosage unit is
\emph{bounded} (content cannot be negative or exceed a physical
maximum) and typically \emph{positively skewed}, especially for
potent low-dose formulations~\cite{orr1978}.  First-principles
models based on statistical mechanics of granular mixtures predict
a parent density that is bounded and unimodal, with shape
controlled by particle-size ratios and drug
loading~\cite{rane2012}.  The Beta family on $[0,1]$---densities
$f(x) = x^{a-1}(1-x)^{b-1}/\mathrm{B}(a,b)$ with $a, b > 0$, where
$\mathrm{B}(a,b) = \int_0^1 t^{a-1}(1-t)^{b-1}\dd{t} = \Gamma(a)\Gamma(b)/\Gamma(a+b)$,
which includes symmetric shapes Beta$(a,a)$, positively skewed shapes
Beta$(a,b)$ with $a < b$, and negatively skewed shapes with
$a > b$---provides a natural parametric
class for modeling these within-batch densities after
normalization to the label-claim interval.  The worked examples below use
Beta$(2,1)$ and Beta$(3,2)$, which are negatively skewed; the variance CDF is unchanged under
the reflection $x\mapsto 1-x$, which swaps $a$ and $b$, so each of them stands equally for its
positively skewed mirror Beta$(1,2)$ and Beta$(2,3)$.

The exact distribution of the sample variance under such
non-Gaussian parents underpins the probability of passing the CU
test, the operating characteristic (OC) curve, and scientifically
justified acceptance limits.  Evaluating those at general $n$ requires
the joint law of $(\bar{X}, s)$, which the companion
engine~\cite{enginecompanion} computes numerically; what closed forms
at $n=3$ add is the exact reference those numerics are validated
against, together with the tail laws and closure obstructions that no
numerical inversion supplies.
With the advent of near-infrared (NIR) spectroscopy for real-time
release testing~\cite{peinado2014, karner2023}, individual-unit
assay data is now available at the scale of hundreds or thousands
of tablets per batch, making it feasible to fit the parent density
directly and apply the closed-form CDFs derived here.

The same mathematical problem---the exact distribution of $s^2$
from a small sample drawn from a bounded, non-Gaussian
parent---arises in several other regulatory and industrial
settings.  Process
capability indices ($C_{pk}$, $P_{pk}$) in semiconductor and
precision manufacturing place $s$ directly in the denominator of
a threshold formula computed from $n = 20$--$50$ measurements of
a bounded quality characteristic, yet standard confidence
intervals assume Gaussian
tails~\cite{cpk2025}.  Reference-scaled average bioequivalence
(RSABE) for highly variable drugs scales the acceptance limits by
$s_{WR}^2$ estimated from $n = 12$--$24$ subjects in a crossover
study, where pharmacokinetic parameters are lognormal and bounded
below~\cite{rsabe2012}.  In each case, the combination of small
$n$, a bounded or asymmetric parent, and a regulatory formula
that involves $s$ or $s^2$ makes the exact distribution of the
sample variance practically consequential.

\paragraph{Why exactness matters: the normal-theory approximation fails at small $n$.}
The default in every setting above is the normal-theory sampling distribution: if the data
were Gaussian with variance $\sigma^2$, then $(n-1)s^2/\sigma^2 \sim \chi^2_{n-1}$, which at
$n = 3$ is $2 s^2/\sigma^2 \sim \chi^2_2$, i.e.\ $F_{\chi^2}(Y) = 1 - e^{-Y/\sigma^2}$.  For a
bounded parent on $[0,1]$ this is not merely inaccurate---it is \emph{qualitatively wrong}.
The exact sample variance is supported on the bounded interval $[0,\tfrac13]$: the maximum
$s^2 = \tfrac13$ is attained at the configuration $\{0,0,1\}$ and is the \emph{same hard
ceiling for every parent} on $[0,1]$, whereas the $\chi^2$ law has unbounded support.
Figure~\ref{fig:chi2-fail} and Table~\ref{tab:chi2-fail} quantify the gap at $n = 3$ (and we
grant the approximation the true $\sigma^2$, its most favorable case):
\begin{itemize}[topsep=2pt,itemsep=1pt]
  \item It assigns positive probability to \emph{impossible} variances:
    $\Pr_{\chi^2}(s^2 > \tfrac13)$ is $1.8\%$ for the uniform parent and $6.9\%$ for the
    arcsine.
  \item Its upper quantiles exceed the maximum attainable variance: the $\chi^2$ $99$th
    percentile of $s^2$ is $0.384$ for the uniform parent---above the ceiling $\tfrac13$---and
    for the arcsine even the $95$th percentile ($0.374$) is unattainable.  The exact $99$th
    percentiles are $0.253$ and $0.312$, so the approximation overstates them by $52\%$ and
    $85\%$.
  \item The Kolmogorov distance $\sup_Y|F - F_{\chi^2}|$---the largest gap between the two
    CDFs at any threshold---reaches $0.064$ (uniform) and $0.103$ (arcsine).  This is the
    metric an acceptance rule cares about: such a rule reports a CDF evaluated at a limit,
    so whatever limit a specification names, the error incurred by using the normal-theory
    law is bounded by this number and, at the worst limit, equals it.  It is also invariant
    under monotone reparametrization, so the same bound holds whether the limit is written
    on $s^2$ or on $s$.  An arcsine parent can therefore put an acceptance probability ten
    percentage points from its true value.
\end{itemize}
At $n = 3$ there is no central-limit rescue, and the error grows with departure from
normality---most severe for the U-shaped arcsine.  Exactness is therefore a requirement, not a
refinement: acceptance probabilities, OC curves, and tolerance limits read off the $\chi^2$ law
can be wrong by tens of percent and can assign risk to variances that cannot physically occur.

\begin{figure}[t]
\centering
\includegraphics[width=\linewidth]{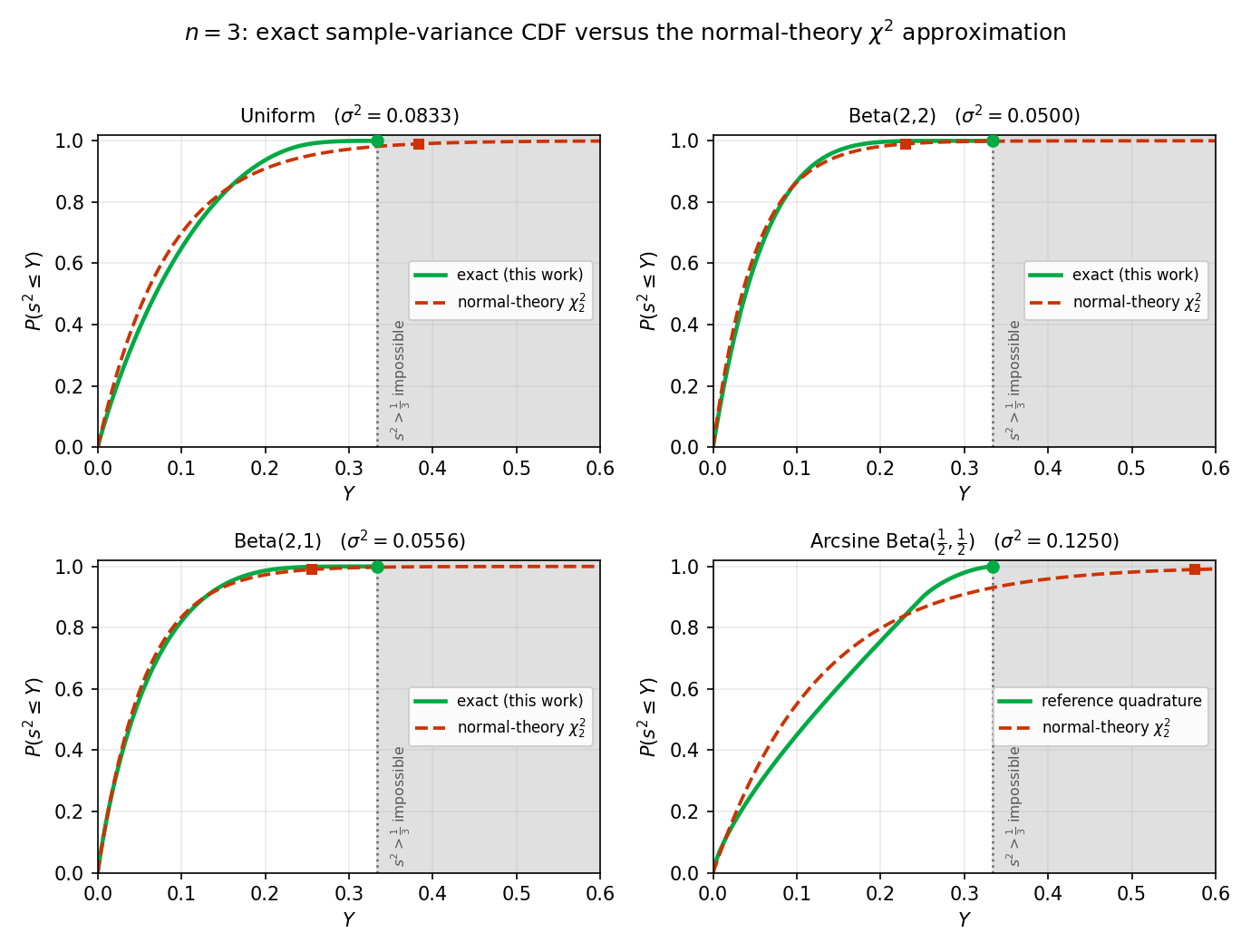}
\caption{Exact $n=3$ sample-variance CDF (green) versus the normal-theory
$\chi^2_2$ approximation (red dashed), for four bounded parents.  \emph{Provenance:} the
uniform, Beta$(2,2)$ and Beta$(2,1)$ curves are the closed forms of this paper,
Eqs.~\eqref{eq:F-uniform-full}, \eqref{eq:beta22-R12}--\eqref{eq:beta22-R3}
and~\eqref{eq:beta21} respectively; the arcsine parent has no closed form, and its curve is the
deterministic quadrature reference of Appendix~\ref{app:arcsine-ref}, labeled accordingly in
the panel.  No curve in this figure is simulated.  The shaded region
$Y > \tfrac13$ is impossible---no sample of three points in $[0,1]$ has $s^2 > \tfrac13$---yet
the $\chi^2$ law rises through it.  The exact curve reaches $1$ at the ceiling
$(\tfrac13,1)$, marked by the green dot, and is constant at $1$ beyond it; only its plotted
support ends there.  The red square marks the $\chi^2$ $99$th percentile; for the uniform
and arcsine parents it lies in the impossible region.  The mismatch worsens with
departure from normality (bottom right, arcsine).}
\label{fig:chi2-fail}
\end{figure}

\begin{table}[h]
\centering
\renewcommand{\arraystretch}{1.2}
\begin{tabular}{lccccc}
\toprule
Parent & $\sigma^2$ & $\sup_Y|F-F_{\chi^2}|$ & exact $q_{.99}$ & $\chi^2$ $q_{.99}$ &
  $\Pr_{\chi^2}(s^2>\tfrac13)$ \\
\midrule
Uniform & $1/12$ & $0.064$ & $0.253$ & $0.384$ & $1.8\%$ \\
Beta$(2,2)$ & $1/20$ & $0.041$ & $0.180$ & $0.230$ & $0.1\%$ \\
Beta$(2,1)$ & $1/18$ & $0.026$ & $0.208$ & $0.256$ & $0.2\%$ \\
Arcsine, Beta$(\tfrac12,\tfrac12)$ & $1/8$ & $0.103$ & $0.312$ & $0.576$ & $6.9\%$ \\
\bottomrule
\end{tabular}
\caption{Failure of the normal-theory ($\chi^2_2$) approximation for the $n=3$ sample variance
across bounded parents.  Exact values are high-accuracy (the uniform row matches the closed
form of Section~\ref{sec:cdf-uniform} to four digits); the $\chi^2$ approximation is granted the
true $\sigma^2$.  Its $99$th percentiles for the uniform and arcsine parents ($0.384$, $0.576$)
exceed the maximum attainable variance $\tfrac13$.}
\label{tab:chi2-fail}
\end{table}

\paragraph{The problem.}
Let $X_1, \ldots, X_n$ be i.i.d.\ from a continuous parent density
$f$ on $[0,1]$, the unit interval being taken without loss of
generality: any bounded interval reduces to it by the affine
rescaling of Remark~\ref{rem:affine}.  The unbiased sample variance
\begin{equation}
  s^2 \;=\; \frac{1}{n-1}\sum_{i=1}^{n}(X_i - \bar{X})^2
  \;=\; \frac{1}{n(n-1)}\sum_{i < j}(X_i - X_j)^2
  \label{eq:s2-def}
\end{equation}
is the fundamental dispersion statistic.  Its exact distribution is
available in closed form for the Gaussian parent at every $n$, as a scaled
chi-squared, and for isolated bounded parents at small $n$, Rietz's uniform $n=3$ result being
the classical instance; no general formula covers arbitrary parents and sample sizes.
For non-Gaussian parents, the constraint surface $\{s^2 = Y\}$ is a
quadric that intersects the support polytope $[0,1]^n$ in a piecewise
geometry whose complexity grows rapidly with $n$.  Even at $n = 3$,
deriving the CDF $F(Y) = \Pr(s^2 \le Y)$ requires resolving this
intersection explicitly.

\paragraph{Prior work.}
Rietz~\cite{rietz1931} gave the closed-form density of the sample
standard deviation for $n = 3$, $f \equiv 1$ (uniform), as a two-piece
piecewise-elementary expression (a parabola joined to an $\arccos$
branch).\footnote{This $n=3$ result
originated as a problem in solid geometry---the area of a diagonal
cylinder's surface lying inside a unit cube---posed by P.~R. Rider
(Problem~3413, \textit{Amer.\ Math.\ Monthly} \textbf{37} (1930), 157)
and solved by Rietz (\textit{ibid.}\ \textbf{38} (1931), 292--295);
Rietz reinterpreted the same integral as a probability in~\cite{rietz1931}.}
The exact sampling distribution of $s^2$ for small $n$ is also treated in
the classical references of Cram\'er~\cite{cramer1946} and Kendall and
Stuart~\cite{kendallStuart1977}; Cram\'er's \S29.3 reproduces the $n = 3$
uniform standard-deviation density.
Ali and Richards~\cite{aliRichards1975} recast the uniform-parent
problem geometrically as the common content of the unit $n$-cube and
a diagonal cylinder, $\Pr(S \le r/\sqrt{n}) = V_n(r,L)/L^n$, and
established the small-radius polynomial form of this volume.
Royen~\cite{royen2007a, royen2007c, royen2007b} obtained exact representations of
the CDF of $s^2$ for gamma and uniform parents at general $n$ via
Laplace transforms and characteristic functions, and
Royen~\cite{royen2008} extended this to any bounded continuous parent
on $[0,1]$: the CDF $F_Q(x) = \Pr(Q \le x)$, where $Q = (n-1)s^2$,
is represented as an infinite Fourier sine series
\begin{equation}
  F_Q(x) = \frac{2}{\pi} \sum_{k=1}^{\infty}
  \frac{\operatorname{Im}\!\bigl(\hat{f}_Q(t_k)\bigr)}{k}
  \bigl(1 - \cos(t_k x)\bigr),
  \qquad t_k = \frac{k\pi}{\sup Q},
  \label{eq:royen-series}
\end{equation}
valid for all $x \in [0, \sup Q]$ and all $n \ge 2$.  Each Fourier
coefficient $\hat{f}_Q(t_k)$ is a closed-form expression (involving
error functions, exponentials, and powers for polynomial parents), but
evaluating $F_Q$ at any particular $x$ requires summing the series to a
tolerance and bounding the truncation error.  This is an exact
\emph{infinite-dimensional} representation; it does not produce finite
piecewise-elementary CDFs.

\paragraph{Contribution.}
We supply the exact $n=3$ distribution that the normal-theory approximation fails to deliver,
for any bounded parent.  The central result is a \emph{closure hierarchy}
(Table~\ref{tab:hierarchy}), organized by the \emph{minimal} function class containing the
parent density---polynomial, rational-but-not-polynomial, algebraic-but-not-rational,
transcendental---so that the rows partition the parents rather than nest.  Its one class-wide
theorem is that polynomial parents on any bounded interval always close in elementary
functions (Theorem~\ref{thm:poly-closure}, stated on $[0,1]$ because the interval is a
normalization---Remark~\ref{rem:affine})---so the entire Beta-polynomial family used to model bounded
data has a finite, exactly computable CDF.  Beyond that, the hierarchy records what is proved
for explicit parents---non-elementarity of the assembled CDF for a rational parent, and of the
radial obstruction integral for the arcsine parent---and what is identified but conjectural
(the exponential and Gaussian rows).  The uniform parent yields a two-piece CDF with a single bifurcation at $Y = 1/4$.
The approach is classical---change coordinates, exploit symmetry, evaluate integrals---and needs
no combinatorial decomposition or symbolic-computation engine.

\paragraph{Outline.}
Section~\ref{sec:coords} introduces the diagonal-orthogonal coordinate
system.  Section~\ref{sec:geometry} describes the hexagonal geometry.
Section~\ref{sec:sextant} builds the sextant chart and the CDF integral,
and records the passage from the variance to the standard deviation.
Section~\ref{sec:cdf-uniform} derives the uniform CDF and the
corresponding law of $S$.
Section~\ref{sec:polynomial} proves the polynomial closure theorem.
Section~\ref{sec:hierarchy} presents the closure hierarchy for
non-polynomial parents.
Section~\ref{sec:singular} treats the singular parents by singularity
isolation, computing them in a handful of terms.
Section~\ref{sec:discussion} collects the discussion: what the
coordinate system buys, what lies beyond the variance marginal, the
extension to $n \ge 4$ and its limits, the acceptance-testing
applications, and the division of labour between exact formulas and
numerical evaluation.

%=====================================================================
\section{The diagonal-orthogonal coordinate system}\label{sec:coords}
%=====================================================================

\subsection{Orthonormal basis}

The main diagonal of $[0,1]^3$ is the direction $\mathbf{d} =
(1,1,1)/\sqrt{3}$.  We choose the orthonormal basis
\begin{equation}
  \mathbf{e}_1 = \frac{1}{\sqrt{2}}(1, -1, 0), \qquad
  \mathbf{e}_2 = \frac{1}{\sqrt{6}}(1, 1, -2), \qquad
  \mathbf{e}_3 = \frac{1}{\sqrt{3}}(1, 1, 1),
  \label{eq:basis}
\end{equation}
and define diagonal-orthogonal coordinates
$(\xi_1, \xi_2, \xi_3) = (X \cdot \mathbf{e}_1,\,
X \cdot \mathbf{e}_2,\, X \cdot \mathbf{e}_3)$:
\begin{equation}
  \xi_1 = \frac{X_1 - X_2}{\sqrt{2}}, \qquad
  \xi_2 = \frac{X_1 + X_2 - 2X_3}{\sqrt{6}}, \qquad
  \xi_3 = \frac{X_1 + X_2 + X_3}{\sqrt{3}} = \sqrt{3}\,\bar{X}.
  \label{eq:xi}
\end{equation}
The inverse transformation is
\begin{equation}
  \begin{pmatrix} X_1 \\ X_2 \\ X_3 \end{pmatrix}
  = \frac{\xi_3}{\sqrt{3}} \begin{pmatrix} 1 \\ 1 \\ 1 \end{pmatrix}
  + \frac{\xi_1}{\sqrt{2}} \begin{pmatrix} 1 \\ -1 \\ 0 \end{pmatrix}
  + \frac{\xi_2}{\sqrt{6}} \begin{pmatrix} 1 \\ 1 \\ -2 \end{pmatrix}.
  \label{eq:inverse}
\end{equation}
The change of variables carries a Jacobian---the determinant of the matrix
of partial derivatives, the factor by which it scales volume elements, so that
$\dd{X_1}\dd{X_2}\dd{X_3} = |J| \dd{\xi_1}\dd{\xi_2}\dd{\xi_3}$.  Here
$|J| = 1$, the map being an orthonormal rotation: volumes, and hence the CDF
integrals below, may be computed in either coordinate system without a correction
factor.

\subsection{Variance as a circle}\label{sec:variance-circle}

Expanding~\eqref{eq:s2-def} in the new coordinates:
\begin{equation}
  \sum_{i=1}^{3}(X_i - \bar{X})^2
  = \xi_1^2 + \xi_2^2,
  \label{eq:parseval}
\end{equation}
by Parseval's theorem (the perpendicular coordinates capture exactly
the deviations from the mean).  Hence
\begin{equation}
  s^2 = \frac{1}{2}(\xi_1^2 + \xi_2^2).
  \label{eq:s2-circle}
\end{equation}
The constraint $s^2 \le Y$ is a \textbf{disk of radius}
$R = \sqrt{2Y}$ in the $(\xi_1, \xi_2)$ plane, independent of
$\xi_3$.  Geometrically, $\{s^2 \le Y\}$ is a \textbf{cylinder of
radius $\sqrt{2Y}$} around the main diagonal, with the
diagonal as its axis.

%=====================================================================
\section{Hexagonal geometry of the cube cross-section}%
\label{sec:geometry}
%=====================================================================

\subsection{Cross-sections perpendicular to the diagonal}

At fixed $\xi_3$ (fixed position along the diagonal), the six cube-face
constraints $0 \le X_i \le 1$ define a convex polygon in the
$(\xi_1, \xi_2)$ plane.  Substituting the
inverse~\eqref{eq:inverse}, the six constraints become:
\begin{alignat}{2}
  X_1 \ge 0 &:\quad \xi_1/\sqrt{2} + \xi_2/\sqrt{6}
    \ge -\xi_3/\sqrt{3}, &\qquad
  X_1 \le 1 &:\quad \xi_1/\sqrt{2} + \xi_2/\sqrt{6}
    \le 1 - \xi_3/\sqrt{3}, \notag\\
  X_2 \ge 0 &:\quad -\xi_1/\sqrt{2} + \xi_2/\sqrt{6}
    \ge -\xi_3/\sqrt{3}, &\qquad
  X_2 \le 1 &:\quad -\xi_1/\sqrt{2} + \xi_2/\sqrt{6}
    \le 1 - \xi_3/\sqrt{3}, \notag\\
  X_3 \ge 0 &:\quad -2\xi_2/\sqrt{6}
    \ge -\xi_3/\sqrt{3}, &\qquad
  X_3 \le 1 &:\quad -2\xi_2/\sqrt{6}
    \le 1 - \xi_3/\sqrt{3}.
  \label{eq:six-faces}
\end{alignat}
Each pair $(X_i \ge 0, X_i \le 1)$ defines a slab of width $1$
perpendicular to one of three directions spaced $120^\circ$ apart in
the $(\xi_1, \xi_2)$ plane.  Their intersection is a convex polygon whose vertex count depends
on the height.  Writing the slice as $X_1+X_2+X_3=3a$, the section is a triangle for
$0<a<\tfrac13$, a hexagon for $\tfrac13<a<\tfrac23$, a triangle again for $\tfrac23<a<1$, and
degenerate at $a\in\{0,\tfrac13,\tfrac23,1\}$: in the outer bands only one constraint of each
pair is active.  The section is a \emph{regular} hexagon only at the midpoint $a=\tfrac12$.  We
write ``polygonal section'' for the general case and reserve ``hexagon'' for the central band,
where the relevant geometry of Sections~\ref{sec:sextant}--\ref{sec:cdf-uniform} lives.  The
head-on projection of the whole cube along the diagonal is a different object, a regular hexagon
at every height, and should not be confused with the sections themselves.

\subsection{Edge distances and symmetry}\label{sec:edge-distances}

Reparametrize the diagonal coordinate as $a = \xi_3/\sqrt{3} \in
[0,1]$, so $a$ is the fractional position along the diagonal ($a = 0$
at vertex $(0,0,0)$, $a = 1$ at $(1,1,1)$).  The six
constraints~\eqref{eq:six-faces} define edges at perpendicular
distances from the origin:
\begin{itemize}[topsep=2pt]
  \item \textbf{Three ``lower'' edges} (from $X_i = 0$): each at
    distance $d_{\rm lo}(a) = a\sqrt{3/2}$ from the origin, with
    outward normals at angles $90^\circ$, $210^\circ$, $330^\circ$.
  \item \textbf{Three ``upper'' edges} (from $X_i = 1$): each at
    distance $d_{\rm hi}(a) = (1-a)\sqrt{3/2}$, normals at
    $30^\circ$, $150^\circ$, $270^\circ$.
\end{itemize}
For $a \le 1/2$: $d_{\rm lo} \le d_{\rm hi}$ (lower edges are closer).
For $a = 1/2$ (cube center): $d_{\rm lo} = d_{\rm hi} = \sqrt{3/8}$
and the hexagon is regular, with
\begin{equation}
  \text{apothem (perpendicular distance from center to edge)}
    = \sqrt{\frac{3}{8}}, \qquad
  \text{circumradius} = \frac{1}{\sqrt{2}}.
  \label{eq:hex-dimensions}
\end{equation}
The larger regular hexagon seen in Figure~\ref{fig:diag-view}b, with
apothem $1/\sqrt{2}$ and circumradius $\sqrt{2/3}$, is the
\emph{shadow} of the whole cube---the union of all cross-sections over
$a$---not any single cross-section.
The hexagon has $3$-fold rotational symmetry (permuting $X_1, X_2, X_3$)
at every $a$, and the cube has an additional $a \leftrightarrow 1 - a$
reflection symmetry ($X_i \to 1 - X_i$).

\subsection{Critical radii and the single bifurcation}%
\label{sec:critical-radii}

The disk of radius $R = \sqrt{2Y}$ is centered at the origin.  As $Y$
increases, $R$ meets two features of the family of cross-sections.  It
first reaches the \emph{edges}: the apothem $d_{\rm lo}(a)$ is largest
at $a = 1/2$, so the disk becomes tangent to the nearest edge there
when $R = \sqrt{3/8}$, i.e.\ $Y = 3/16$; Figure~\ref{fig:xsection} draws
the apothem and the circumradius side by side.  This tangency creates
\emph{no} new regime---the shaving integral of \S\ref{sec:R12} absorbs
it smoothly.  The bifurcation comes instead from the \emph{vertices}:
the cross-section circumradius $\sqrt{2(1 - 3a + 3a^2)}$ is minimized
at $a = 1/2$, where it equals $1/\sqrt{2}$, so
\begin{equation}
  R = \frac{1}{\sqrt{2}} \qquad
  \Longleftrightarrow \qquad Y = \frac{1}{4}
  \label{eq:Y-crit}
\end{equation}
is the first $Y$ at which the disk crosses a cross-section vertex---the
green circle of Figure~\ref{fig:diag-view}b.
This is the \textbf{only geometric bifurcation}:
\begin{itemize}[topsep=2pt]
  \item $Y < 1/4$: at every angle $\theta$ there is a band of heights
    $z$ on which the variance disk is the binding constraint, and the
    shaving produced by the lower and upper facets proceeds
    independently (Figure~\ref{fig:three-regions}, left).
  \item $Y > 1/4$: near the vertex direction the variance band pinches
    off---the lower- and upper-facet regions meet directly---and the
    disk crosses the hexagon vertices around the cube center.  The CDF
    integral acquires new geometric terms
    (Figure~\ref{fig:three-regions}, right).
\end{itemize}

\begin{figure}[H]
\centering
\includegraphics[width=\linewidth]{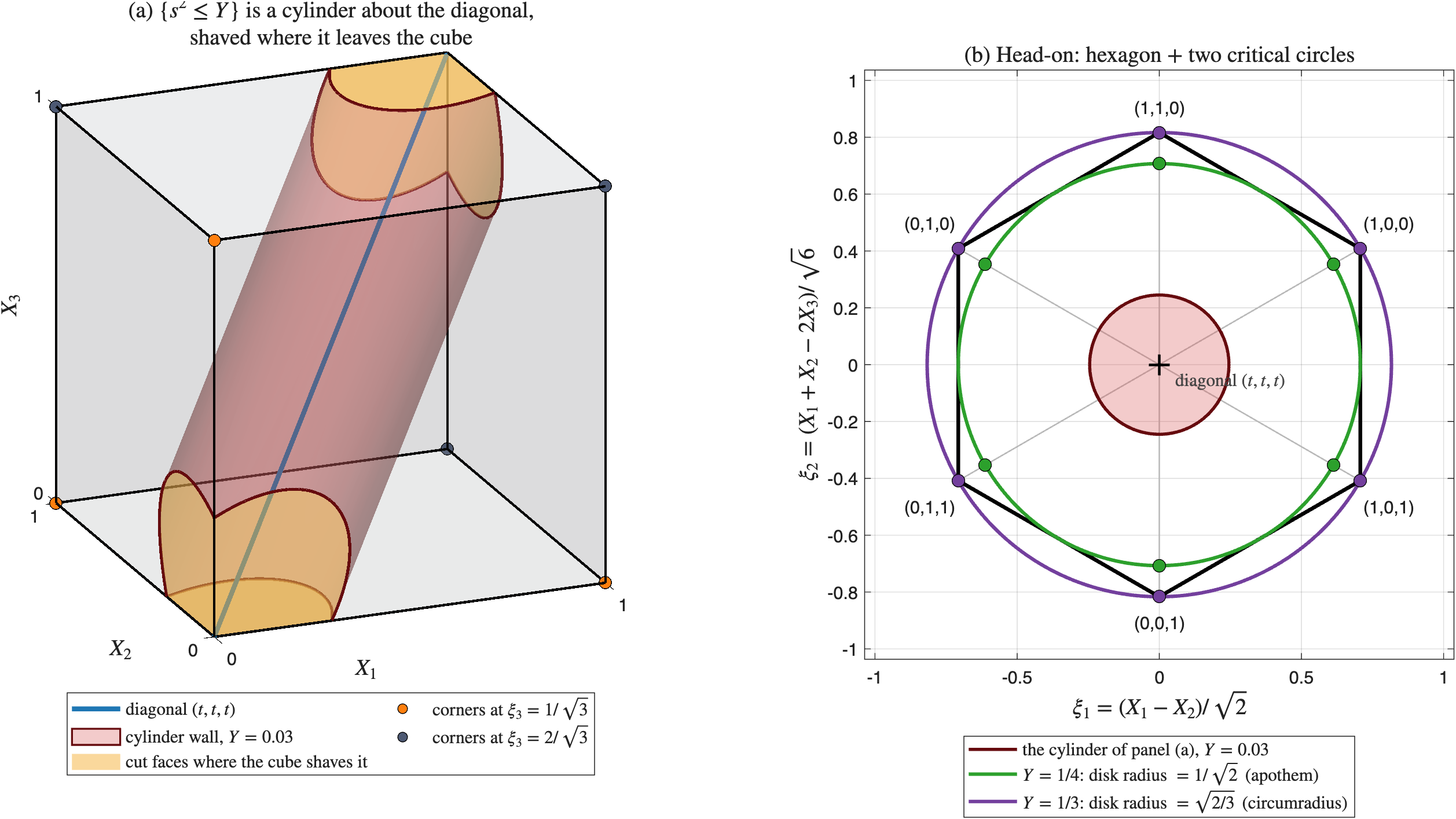}
\caption{The same geometry in three dimensions and seen down the diagonal.
\emph{(a)} In the cube, $\{s^2 \le Y\}$ is a solid cylinder of radius
$\sqrt{2Y}$ about the diagonal $(t,t,t)$, drawn here at $Y = 0.03$
with a transparent wall so the interior is visible.  Near the two end
corners the cylinder runs out of the cube and is \emph{shaved} by the
facets; the amber regions are the resulting cut faces, bounded by the
exact facet--cylinder intersection arcs (dark red).  This is the
decomposition used in \S\ref{sec:cyl-minus-shave}.  The six remaining
corners sit at \emph{two} different heights along the diagonal,
$\xi_3 = 1/\sqrt{3}$ (orange) and $2/\sqrt{3}$ (slate).
\emph{(b)} Looking down the diagonal, the cube collapses to a
regular hexagon (apothem $1/\sqrt{2}$, circumradius $\sqrt{2/3}$)
and the cylinder to a disk; the two rings of corners in (a) collapse
onto the six hexagon vertices, labeled by their cube-corner origin.
The cylinder of panel~(a) reappears here as the small filled disk, in
the same colors.  Two further circles mark the critical radii:
$Y = 1/4$ (green, disk reaches the shadow-hexagon apothem,
equivalently the nearest cross-section vertices---the \textbf{single
bifurcation}), $Y = 1/3$ (violet, disk reaches the shadow-hexagon
vertices and $F = 1$).}
\label{fig:diag-view}
\end{figure}

\begin{figure}[t]
\centering
\includegraphics[width=\linewidth]{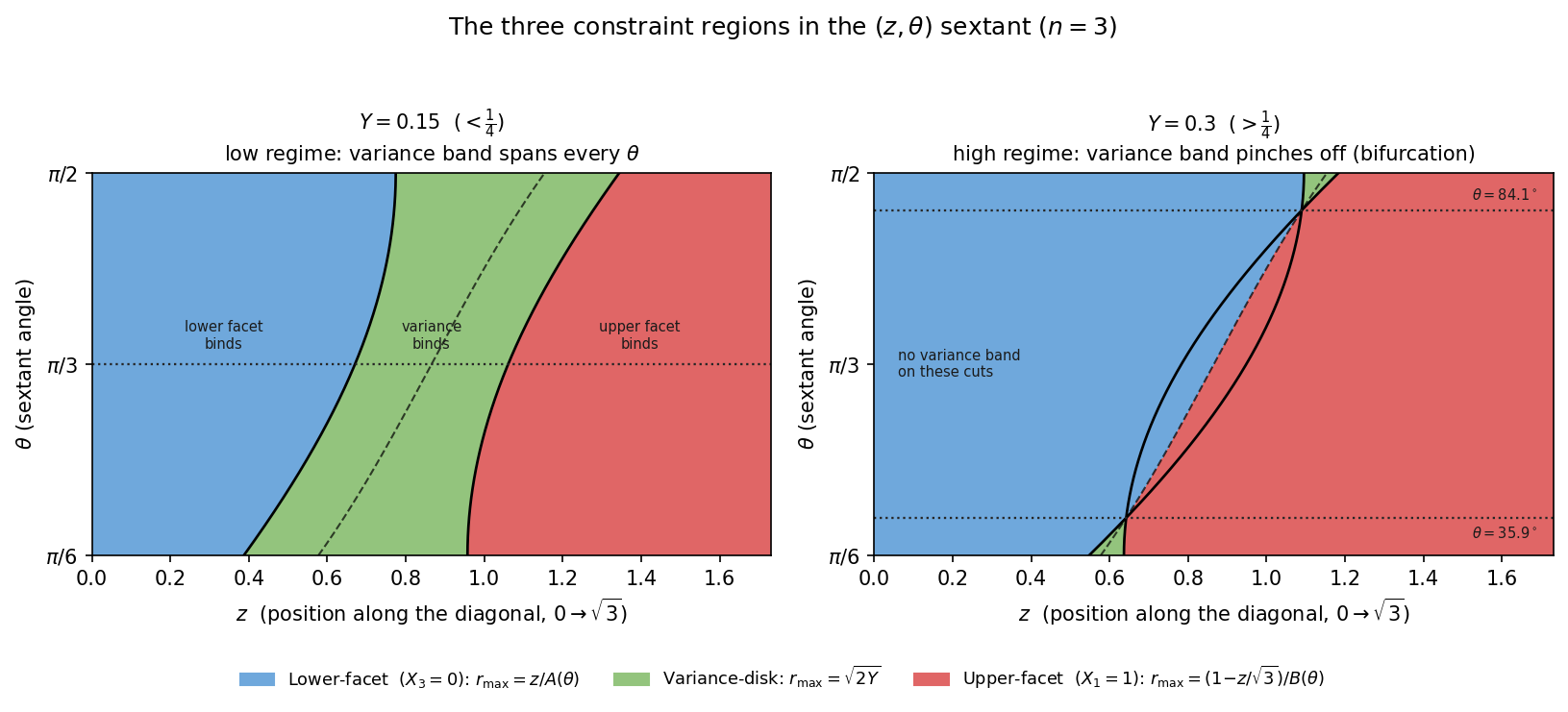}
\caption{The three constraint regions of~\eqref{eq:rmax} over the sextant
$(z,\theta)\in[0,\sqrt3]\times[\tfrac\pi6,\tfrac\pi2]$, colored by which
constraint sets $r_{\max}$.  \emph{How to read it:} this is parameter space, not
physical space.  Fix an angle $\theta$---a \emph{horizontal} cut---and sweep the height
$z$ from $0$ to $\sqrt3$: the binding constraint changes blue $\to$ green $\to$ red,
which is the sequence of cross-sections drawn in Figure~\ref{fig:xsection}, one panel
per color.  \emph{Left} ($Y=0.15<\tfrac14$): the
variance-disk band (green, $r_{\max}=\sqrt{2Y}$) separates the lower-facet
band (blue, $r_{\max}=z/A(\theta)$) from the upper-facet band (red,
$r_{\max}=(1-z/\sqrt3)/B(\theta)$), with the facet coefficients $A, B$
of~\eqref{eq:AB}, at \emph{every} $\theta$; the dotted cut at
$\theta=\pi/3$ is labeled with the three constraints in the order they bind.
\emph{Right} ($Y=0.30>\tfrac14$): the variance band pinches off over an
\emph{interior} range of angles straddling the vertex direction $\theta=\pi/3$---here
$35.9^\circ$ to $84.1^\circ$, marked by the two dotted cuts---on which no height $z$
leaves the disk inside the cross-section, so the green band is absent from those cuts
entirely.  It survives at both ends of the sextant, which is why a green sliver remains
near $\theta=\pi/6$.  There the disk protrudes past both facets and the lower- and
upper-facet regions meet directly along $z_{\rm cross}$ (dashed), $z_{\rm cross}(\theta)
= A(\theta)/\!\bigl(B(\theta)+A(\theta)/\sqrt3\bigr)$.  This vanishing of the
variance band is the geometric bifurcation at $Y=\tfrac14$.  Solid curves are
the region boundaries $z=A(\theta)\sqrt{2Y}$ and $z=\sqrt3(1-B(\theta)\sqrt{2Y})$.}
\label{fig:three-regions}
\end{figure}

The maximum variance is $Y_{\max} = 1/3$, at which the disk radius
$R = \sqrt{2/3}$ equals the distance from the axis to the six
off-diagonal cube corners (the shadow-hexagon circumradius of
Figure~\ref{fig:diag-view}b): the cylinder swallows the entire cube,
and $F(1/3) = 1$.

%=====================================================================
\section{The sextant chart and the CDF integral}\label{sec:sextant}
%=====================================================================

The hexagon's $3$-fold rotational symmetry, combined with the $S_3$
permutation symmetry of i.i.d.\ sampling, suggests working in a single
$60^\circ$ wedge.  We introduce cylindrical coordinates
$(r, \theta, z)$ around the diagonal:
\begin{equation}
  z = \xi_3 = \frac{X_1 + X_2 + X_3}{\sqrt{3}}, \qquad
  r = \sqrt{\xi_1^2 + \xi_2^2} = \sqrt{2\,s^2}, \qquad
  \theta = \operatorname{atan2}(\xi_2, \xi_1).
  \label{eq:cyl-coords}
\end{equation}
Here the Jacobian is $|J| = r$, so the volume element is
$r\dd{r}\dd{\theta}\dd{z}$.  That factor of $r$ is what makes the inner
integral elementary: $\int_0^{r_{\max}} r \dd{r} = \tfrac12 r_{\max}^2$, which is
why a uniform parent needs no radial antiderivative at all.

\subsection{The ordered sextant}

Restricting to the ordered sextant $X_1 \ge X_2 \ge X_3$ confines
$\theta$ to the $60^\circ$ wedge $\theta \in [\pi/6, \pi/2]$
(Figure~\ref{fig:sextant}).  In
this wedge, only \textbf{two cube facets are active}:
Write
\begin{equation}
  A(\theta) = \sqrt{2}\sin\theta, \qquad
  B(\theta) = \frac{\cos\theta}{\sqrt{2}} + \frac{\sin\theta}{\sqrt{6}}
             = \tfrac{2}{\sqrt{6}}\sin(\theta + \pi/3)
  \label{eq:AB}
\end{equation}
for the two facet coefficients, used throughout.  Then:
\begin{itemize}[topsep=2pt]
  \item $X_3 = 0$ (lower facet), at distance
    $r_{\rm lo}(z, \theta) = z / A(\theta)$,
  \item $X_1 = 1$ (upper facet), at distance
    $r_{\rm hi}(z, \theta) = (1 - z/\sqrt{3}) / B(\theta)$.
\end{itemize}

\begin{figure}[H]
\centering
\includegraphics[width=\linewidth]{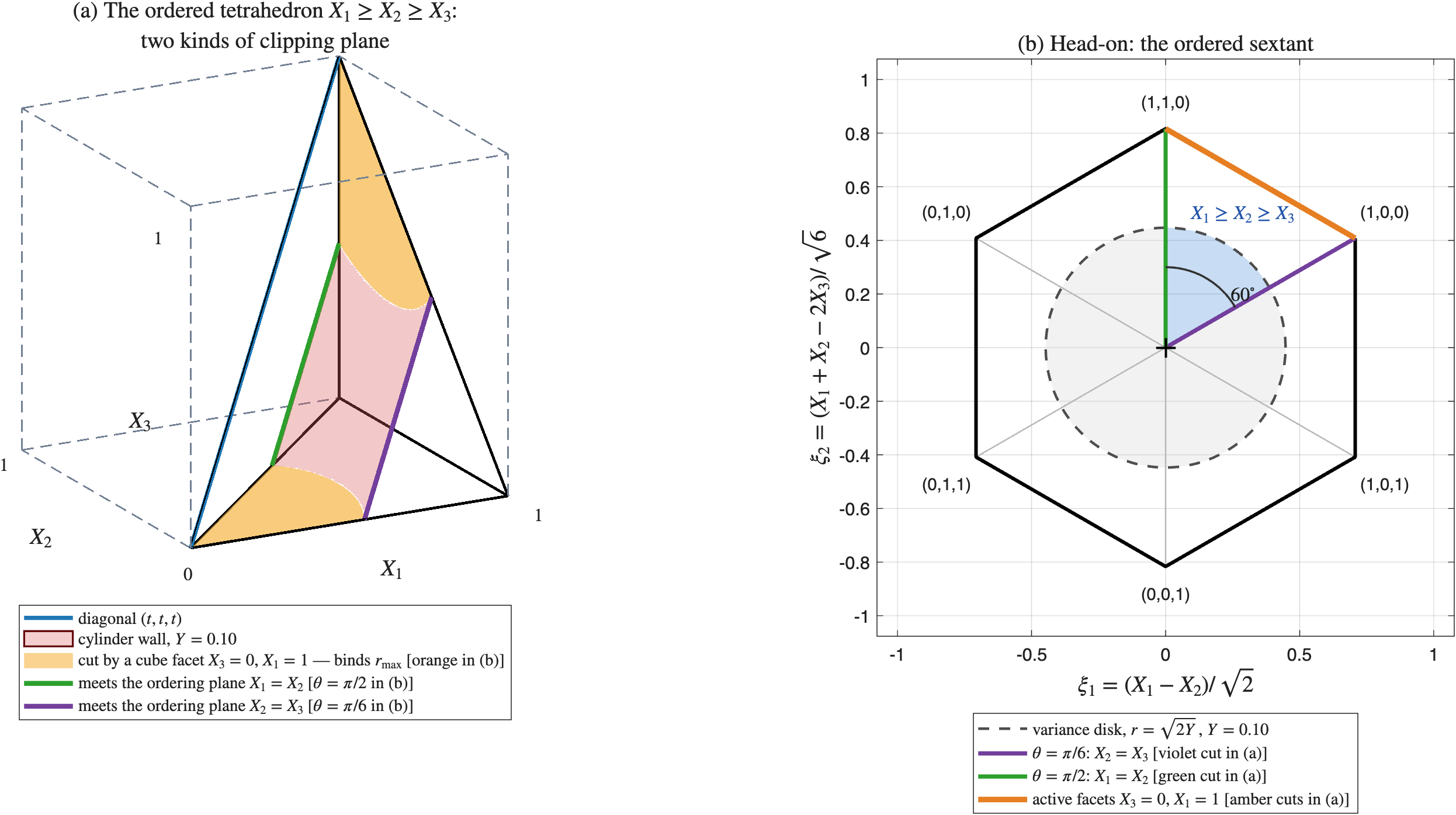}
\caption{The ordered region in three dimensions and seen down the diagonal.
\emph{(a)} $X_1 \ge X_2 \ge X_3$ is one of six congruent tetrahedra
tiling the cube, here
$\mathrm{conv}\{(0,0,0),(1,0,0),(1,1,0),(1,1,1)\}$.  The variance
cylinder restricted to it (transparent wall, $Y = 0.10$) is clipped by
four planes of two distinct kinds.  The cube facets $X_3 = 0$ and
$X_1 = 1$ are genuine boundaries of the support and the only ones
entering $r_{\max}$; the cylinder is cut open there, and those cut
faces are drawn filled (amber).  The ordering planes $X_1 = X_2$ and
$X_2 = X_3$ are mere symmetry cuts---the five other tetrahedra sit on
their far sides---so the cylinder merely meets them, along the straight
lines $|u - v| = r\sqrt{3/2}$ drawn in green and violet.
Distinguishing the two kinds is what makes the sextant reduction
legitimate.  \emph{The colors carry over to panel~(b):} green
$\leftrightarrow \theta = \pi/2$, violet $\leftrightarrow
\theta = \pi/6$, amber $\leftrightarrow$ the orange active-facet edge.
\emph{(b)} The same region seen down the diagonal: the $60^\circ$
wedge $\theta \in [\pi/6, \pi/2]$ of the hexagonal cross-section.  The
boundary $\theta = \pi/6$ (violet) is where $X_2 = X_3$;
$\theta = \pi/2$ (green) is where $X_1 = X_2$.  The single hexagon edge inside the
wedge (orange) is the common image of the two active facets $X_3 = 0$
and $X_1 = 1$, whose constraints define $r_{\rm lo}$ and $r_{\rm hi}$
respectively.  Dashed: the variance disk $r = \sqrt{2Y}$.  The factor
of $6$ in the CDF integral accounts for the five remaining sextants
(gray).}
\label{fig:sextant}
\end{figure}

\subsection{The CDF as a single integral}

The CDF for any parent density $f$ is
\begin{equation}
  F(Y) = 6 \int_0^{\sqrt{3}} dz
    \int_{\pi/6}^{\pi/2} d\theta
    \int_0^{r_{\max}} r\, f(X_1)\,f(X_2)\,f(X_3) \dd{r},
  \label{eq:CDF-sextant}
\end{equation}
where
\begin{equation}
  r_{\max}(z, \theta, Y) = \min\bigl(\sqrt{2Y},\;
    r_{\rm lo}(z, \theta),\; r_{\rm hi}(z, \theta)\bigr),
  \label{eq:rmax}
\end{equation}
and $X_i = X_i(\xi_1, \xi_2, \xi_3) = X_i(r, \theta, z)$ via the
inverse transformation~\eqref{eq:inverse} with $\xi_1 = r\cos\theta$,
$\xi_2 = r\sin\theta$.

The factor of $6$ accounts for the $S_3$ permutation symmetry.  For
parents with the additional reflection symmetry $f(x) = f(1-x)$
(e.g., uniform, Beta$(a,a)$), the cube's $Z_2$ symmetry
$X_i \to 1 - X_i$ (which maps $z \to \sqrt{3} - z$, i.e., $a \to 1-a$)
halves the $z$-range:
\begin{equation}
  F(Y) = 12 \int_0^{\sqrt{3}/2} dz
    \int_{\pi/6}^{\pi/2} d\theta
    \int_0^{r_{\max}} r\, f(X_1)\,f(X_2)\,f(X_3) \dd{r}.
  \label{eq:CDF-sextant-sym}
\end{equation}

\paragraph{From the variance to the standard deviation.}
Acceptance criteria are normally written in terms of the sample
standard deviation $S = \sqrt{s^2}$---for instance the USP
$\langle 905\rangle$ acceptance value
$\mathrm{AV} = |M - \bar{X}| + k\,S$ used in
\S\ref{sec:apply}---rather than in terms of $s^2$.  Passing from one to
the other costs no further integration, so every closed form obtained
below carries over.

\begin{proposition}[Standard-deviation law]\label{prop:variance-to-sd}
Let $F(Y) = \Pr(s^2 \le Y)$ be the CDF of~\eqref{eq:CDF-sextant}, and
$Y_{\max} = \sup\{Y : F(Y) < 1\}$.  Since $t \mapsto t^2$ is a strictly
increasing bijection of $[0,\infty)$, the sample standard deviation
$S = \sqrt{s^2}$ has
\begin{equation}
  G(t) \;:=\; \Pr(S \le t) \;=\; \Pr(s^2 \le t^2) \;=\; F(t^2),
  \label{eq:sd-cdf}
\end{equation}
and, wherever $F$ is differentiable, density
\begin{equation}
  g(t) \;=\; 2t\,F'(t^2),
  \label{eq:sd-pdf}
\end{equation}
on the support $t \in \bigl(0, \sqrt{Y_{\max}}\,\bigr]$.  In particular
$G$ and $g$ lie in the same function class as $F$, since $t \mapsto t^2$
is algebraic.
\end{proposition}

\noindent
For the $n = 3$ uniform parent $Y_{\max} = 1/3$, so $S$ ranges over
$\bigl(0, 1/\sqrt{3}\,\bigr]$.  Every closed-form variance CDF obtained
in this paper therefore yields the corresponding standard-deviation CDF
and density in closed form, exactly and without truncation.

\subsection{Three constraint regions}\label{sec:three-regions}

At each $(z, \theta)$, exactly one of three constraints in
\eqref{eq:rmax} is binding, in terms of the facet
coefficients~\eqref{eq:AB}:
\begin{enumerate}[topsep=2pt]
  \item \textbf{Lower-facet region} ($X_3 = 0$ binds):
    $z < A(\theta)\sqrt{2Y}$ and $z < z_{\rm cross}(\theta)$.
    Here $r_{\max} = z/A(\theta)$.
  \item \textbf{Variance-disk region} (variance binds):
    $z \ge A(\theta)\sqrt{2Y}$ and
    $(1 - z/\sqrt{3}) \ge B(\theta)\sqrt{2Y}$.
    Here $r_{\max} = \sqrt{2Y}$.
  \item \textbf{Upper-facet region} ($X_1 = 1$ binds):
    $(1 - z/\sqrt{3}) < B(\theta)\sqrt{2Y}$ and $z > z_{\rm cross}$.
    Here $r_{\max} = (1 - z/\sqrt{3})/B(\theta)$.
\end{enumerate}
The boundary between the lower-facet and variance-disk regions is
$z = A(\theta)\sqrt{2Y}$; the boundary between the variance-disk and
upper-facet regions is $z = \sqrt{3}(1 - B(\theta)\sqrt{2Y})$.

Geometrically (Figure~\ref{fig:xsection}), slide a plane perpendicular to the
cube diagonal along it.  The cube's cross-section is a small \emph{triangle}
near each corner and a \emph{hexagon} in the middle, while the variance
constraint is a \emph{disk} of fixed radius $\sqrt{2Y}$ centered on the axis.
The integration radius $r$ runs out to whichever boundary is closer: near the
bottom corner (small $z$) the triangle lies inside the disk, so the cube facet
$X_3=0$ binds; in the middle the disk lies inside the hexagon, so the variance
binds; near the top corner (large $z$) the upper triangle lies inside the
disk, so $X_1=1$ binds.  The disk size is fixed by $Y$; the cube slice grows
then shrinks, so the binding boundary switches cube\,$\to$\,disk\,$\to$\,cube.

For $Y \le 1/4$, the variance-disk region spans a nonempty interval
of $z$ at every $\theta \in [\pi/6, \pi/2]$.  Two distinct thresholds must be kept apart here.
Full containment of the disk in the central hexagon holds only for $Y \le 3/16$, where the disk
radius $\sqrt{2Y}$ first reaches the apothem $\sqrt{3/8}$; beyond it the disk protrudes through
the edges and is shaved.  That transition is absorbed smoothly by the shaving integral of
\S\ref{sec:R12} and introduces no new analytic branch, which is why $Y=\tfrac14$ remains the only
\emph{true} bifurcation.  What persists up to $Y=\tfrac14$ is the weaker statement above, that the
variance-binding band in $z$ is nonempty at every angle, not containment.

For $Y > 1/4$, the variance-disk region may vanish at some $\theta$
values (the disk protrudes past both active facets near the vertex
direction $\theta = \pi/3$), creating the
geometric bifurcation.  Figure~\ref{fig:three-regions} maps the three
regions over the $(z,\theta)$ sextant; Figure~\ref{fig:cyl-minus-shaved}
illustrates the disk--hexagon intersection at three cross-sections.

%=====================================================================
\section{Derivation for the uniform parent}\label{sec:cdf-uniform}

\noindent\emph{Integrand-level visualizations of every zone of this derivation are
collected in Appendix~\ref{app:atlas}; the main text keeps only the two structural
figures.}
%=====================================================================

With $f \equiv 1$, the inner integral in~\eqref{eq:CDF-sextant} is
$\int_0^{r_{\max}} r \dd{r} = \tfrac{1}{2}r_{\max}^2$, so
\begin{equation}
  F(Y) = 3 \int_0^{\sqrt{3}} dz
    \int_{\pi/6}^{\pi/2} r_{\max}^2(z, \theta, Y) \dd{\theta}.
  \label{eq:CDF-unif}
\end{equation}

\subsection{Cylinder minus shaving}\label{sec:cyl-minus-shave}

The CDF equals the volume of the intersection of the variance
cylinder with the unit cube.  We decompose this as
\begin{equation}
  F(Y) \;=\; V_{\rm cyl}(Y) \;-\; V_{\rm shaved}(Y),
  \label{eq:cyl-minus-shave}
\end{equation}
where $V_{\rm cyl}$ is the volume of the full cylinder (radius
$\sqrt{2Y}$, height $\sqrt{3}$ along the diagonal), and
$V_{\rm shaved}$ is the volume where the cylinder protrudes past the
cube faces.

\begin{figure}[H]
\centering
\includegraphics[width=\linewidth]{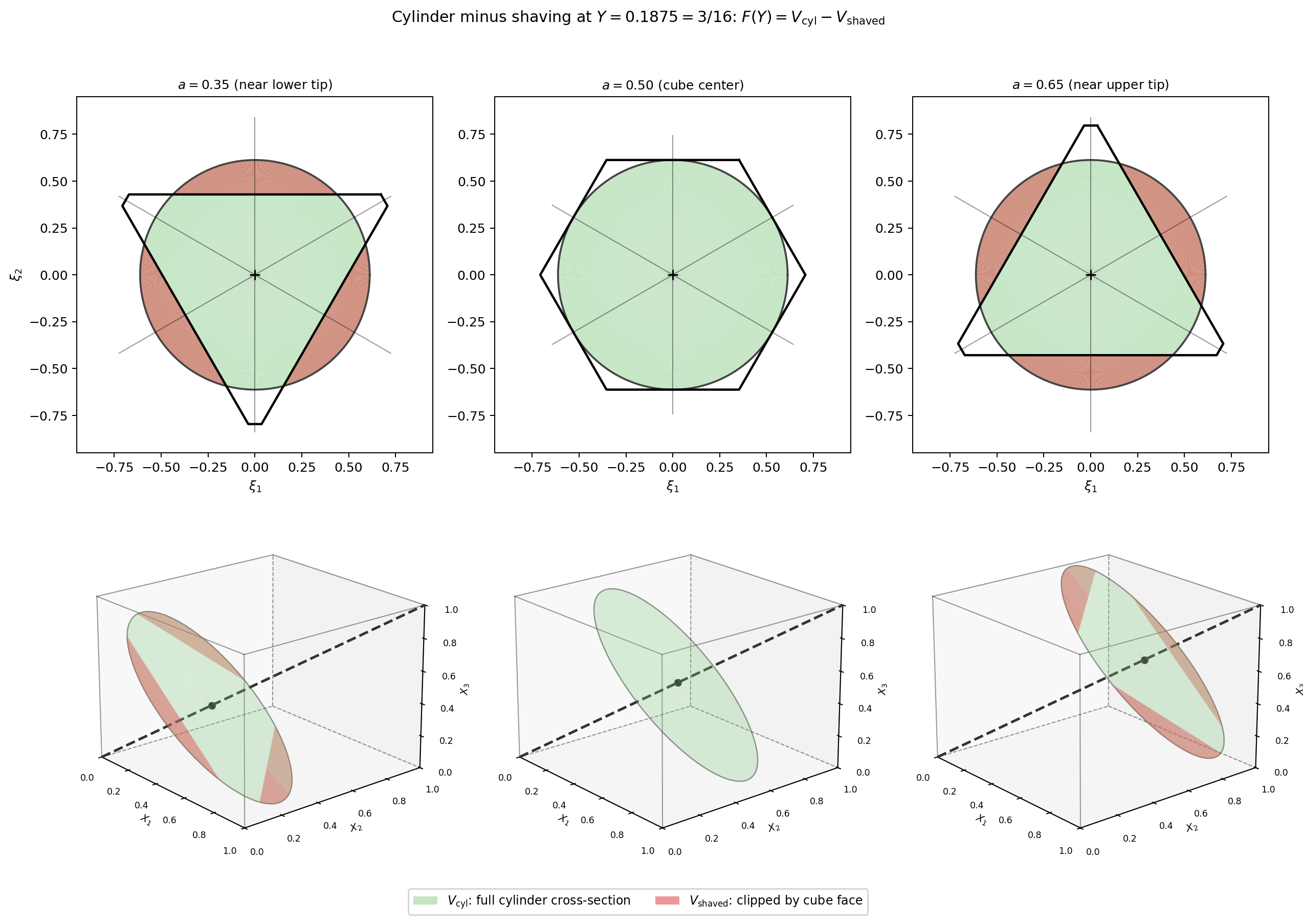}
\caption{The cylinder-minus-shaving decomposition at $Y = 3/16$.
Green = full cylinder cross-section ($V_{\rm cyl}$).
Red = volume shaved off by cube faces ($V_{\rm shaved}$).
The CDF is the green area minus the red area at each height,
integrated along the diagonal.
At $a = 0.35$ and $a = 0.65$, the red segments show where
the cube clips the cylinder; at $a = 0.50$ the disk just
touches the hexagon (no shaving).}
\label{fig:cyl-minus-shaved}
\end{figure}

The cylinder volume is immediate:
\begin{equation}
  V_{\rm cyl} = \int_0^{\sqrt{3}} \pi R^2 \dd{z}
    = \sqrt{3}\cdot \pi \cdot 2Y = 2\sqrt{3}\,\pi\,Y.
  \label{eq:Vcyl}
\end{equation}

The shaved volume is the integral of the annular excess at each
$(z, \theta)$ where the disk extends past the hexagon:
\begin{equation}
  V_{\rm shaved} = 3\int_0^{\sqrt{3}} dz
    \int_{\pi/6}^{\pi/2}
      \bigl(2Y - r_{\max}^2(z,\theta,Y)\bigr) \dd{\theta}.
  \label{eq:Vshaved}
\end{equation}
In the sextant, only two cube faces are active: $X_3 = 0$ (lower)
and $X_1 = 1$ (upper).  Each contributes independently to the
shaving.  For the uniform parent, the $a \leftrightarrow 1-a$
symmetry gives $V_{\rm shaved}^{\rm hi} = V_{\rm shaved}^{\rm lo}$,
so
\begin{equation}
  V_{\rm shaved} = 2\,V_{\rm shaved}^{\rm lo}.
  \label{eq:shave-sym}
\end{equation}

\subsection{$Y \le 1/4$: thin cylinder inside the cube}\label{sec:R12}

The lower face $X_3 = 0$ clips the disk whenever
$z < A(\theta)\sqrt{2Y}$, replacing $r_{\max} = \sqrt{2Y}$ with
$r_{\max} = z/A(\theta)$, where $A(\theta) = \sqrt{2}\sin\theta$.
At a fixed height~$z$, the clipping condition
$\sin\theta > z/(2\sqrt{Y})$ restricts the angular range to
$\theta \in [\theta_{\min}(z),\, \pi/2]$, where
\begin{equation}
  \theta_{\min}(z) = \max\!\bigl(\tfrac{\pi}{6},\;
    \arcsin\bigl(z\,/\,2\sqrt{Y}\bigr)\bigr).
  \label{eq:theta-min}
\end{equation}
For $z \le \sqrt{Y}$: $\theta_{\min} = \pi/6$ and the shaving
fills the full sextant (Figure~\ref{fig:shaved-lo}a).
For $z > \sqrt{Y}$: the angular range narrows toward $\pi/2$
(Figure~\ref{fig:shaved-lo}b).  The shaved volume from this face
is the integral of these red annular areas over all heights~$z$:
\begin{equation}
  V_{\rm shaved}^{\rm lo}
  = 3\int_0^{2\sqrt{Y}} dz
    \int_{\theta_{\min}(z)}^{\pi/2}
      \Bigl(2Y - \frac{z^2}{A(\theta)^2}\Bigr) d\theta.
  \label{eq:Vshaved-lo-zform}
\end{equation}

Swapping the integration order in~\eqref{eq:Vshaved-lo-zform}
(each fixed~$\theta$ aggregates contributions from $z = 0$ to
$z = A(\theta)\sqrt{2Y}$) puts the $\theta$-limits at $\pi/6$ and
$\pi/2$, and with $A(\theta) = \sqrt{2}\sin\theta$ both integrals are
elementary:
\begin{equation}
  V_{\rm shaved}^{\rm lo}
  = 3\int_{\pi/6}^{\pi/2} \!d\theta
     \int_0^{A(\theta)\sqrt{2Y}}\!\!
       \Bigl(2Y - \frac{z^2}{A(\theta)^2}\Bigr) dz
  = 4\sqrt{3}\,Y^{3/2}.
  \label{eq:Vshaved-lo}
\end{equation}

Assembling~\eqref{eq:cyl-minus-shave}--\eqref{eq:Vshaved-lo}:
\begin{equation}
  F(Y)
    = \underbrace{2\sqrt{3}\,\pi\,Y}_{V_{\rm cyl}}
    - \underbrace{8\sqrt{3}\,Y^{3/2}}_{V_{\rm shaved}}
    = 2\sqrt{3}\,Y\bigl(\pi - 4\sqrt{Y}\bigr),
    \qquad 0 \le Y \le \tfrac{1}{4}.
  \label{eq:FR12}
\end{equation}

\paragraph{Verification.}
Two properties of~\eqref{eq:FR12} are immediate: $F(0) = 0$, as any
distribution function must satisfy, and
$F'(Y) = 2\sqrt{3}\,(\pi - 6\sqrt{Y})$, which is positive for
$Y < \pi^{2}/36 \approx 0.274$ and therefore makes $F$ strictly
increasing across the whole of $[0, \tfrac14]$.

The substantive test is external: passing~\eqref{eq:FR12} to the
standard deviation gives the density $g(t) = 4\sqrt{3}\,t(\pi - 6t)$,
precisely the parabola Rietz~\cite{rietz1931} obtained in 1931 once his
divisor-$n$ normalization is rescaled to the unbiased one used here
(Remark~\ref{rem:rietz} and \S\ref{sec:variance-to-sd}); the lower
regime \emph{reproduces} the classical result rather than merely
agreeing with it.  The endpoint value
$F(\tfrac14) = \tfrac{\sqrt{3}}{2}(\pi - 2) \approx 0.9886$ is by
contrast no test---it is~\eqref{eq:FR12} evaluated at $Y = \tfrac14$,
and nothing derived so far constrains it---but a prediction: the
$Y > \tfrac14$ branch of \S\ref{sec:R3} follows from a different
decomposition of the shaving integral, and continuity forces the two to
agree at the bifurcation.  That they do checks both, as does the
normalization anchor $F(\tfrac13) = 1$ on that upper branch.

\paragraph{Geometric interpretation.}
The cylinder volume $2\sqrt{3}\pi Y$ grows linearly in $Y$.
The shaving correction $8\sqrt{3}Y^{3/2}$ grows as $Y^{3/2}$: at
small~$Y$ the cylinder is thin and the cube barely clips it; the
shaving is a higher-order correction to the cylinder volume.
At $Y = 3/16$, the disk just touches the hexagon apothem at the
widest cross-section ($a = 1/2$), but this does not create a new
regime---the shaving integral~\eqref{eq:Vshaved-lo} absorbs
it smoothly.

\subsection{$Y > 1/4$: wide cylinder with cube corners protruding}\label{sec:R3}

At $Y = 1/4$ the disk radius $\sqrt{2Y} = 1/\sqrt{2}$ equals the
circumradius of the regular hexagon at $a = 1/2$: the disk fully
circumscribes the hexagon at the cube center.  For $Y > 1/4$, the
shaving correction acquires new structure---a band of $z$-values near
$\sqrt{3}/2$ where the disk fully contains the hexagon, so the
\emph{entire} hexagonal cross-section contributes and the shaving
saturates.

\paragraph{Shaving structure.}
Define $\Delta = \sqrt{4Y - 1}$ (real for $Y > 1/4$).  Moving along
the diagonal from a cube corner ($a = 0$) toward the cube center
($a = 1/2$), the cross-section grows and its relationship to the
disk passes through three zones (Figure~\ref{fig:regime3-zones}):
\begin{enumerate}[topsep=2pt]
  \item \textbf{Tip zone} ($a < a_v := \sqrt{Y/3}$): the cross-section
    is a small triangle near a cube corner, its vertices at distance
    $a\sqrt{6}$ from the axis.  For $a < a_v$ the disk fully contains
    it; shaving $=$ full annulus (disk area $-$ triangle area).
  \item \textbf{Partial-clip zone} ($a_v < a < a_*$, where
    $a_* = \tfrac{1}{2}(1 - \Delta/\sqrt{3})$): the vertices of the
    cross-section protrude past the disk and are clipped by the disk
    boundary into circular arcs.  Two things happen inside this zone
    without ending it: at $a = 1/3$ the triangle--hexagon transition
    introduces the three $X_i = 1$ edges at the protruding corners,
    and at $a_u := 1 - 2\sqrt{Y/3}$ those upper edges begin to be
    clipped as well.  As $a$ increases further, the circumradius
    $\sqrt{2(1 - 3a + 3a^2)}$ decreases; the arcs shrink until at
    $a = a_*$ the vertices land on the disk circle and the arcs
    vanish.
  \item \textbf{Full-saturation zone} ($a_* < a < 1 - a_*$):
    the disk fully contains the hexagon; all edges are straight.
    Shaving $=$ full annulus (disk area $-$ hexagon area),
    independent of $Y$.
\end{enumerate}
The upper half of the diagonal ($a > 1/2$) is the mirror image
by $a \leftrightarrow 1-a$.

\paragraph{Zone-by-zone shaving.}
Using the $a \leftrightarrow 1-a$ symmetry, we integrate over the
lower half of the diagonal ($0 \le z \le \sqrt{3}/2$) and double.
In the height variable the zone boundaries are
$z_v = \sqrt{3}\,a_v = \sqrt{Y}$,
$z_u = \sqrt{3}\,a_u = \sqrt{3} - 2\sqrt{Y}$, and
$z_* = \sqrt{3}\,a_*$ (where the circumradius equals $R$), ordered
$0 < z_v < z_u < z_* < \sqrt{3}/2$ for $1/4 < Y < 1/3$.
The total shaving is
\begin{equation}
  V_{\rm shaved} = 2\bigl(S_{\rm tip} + S_{\rm clip} + S_{\rm sat}\bigr),
  \label{eq:Vshaved-split}
\end{equation}
with one integral per zone:

\paragraph{1.\ Tip-zone shaving ($0 < z < z_v$).}
The cross-section is a triangle fully inside the disk
(Figure~\ref{fig:zone-tip}).  In the
sextant, $r_{\rm tri}(z,\theta) = d_{\rm lo}/\sin\theta$ with
$d_{\rm lo} = z/\sqrt{2}$, so $r_{\rm tri}^2 = z^2/(2\sin^2\theta)$.
The shaving equals the full annulus:
\begin{equation}
  S_{\rm tip} = 3\int_0^{z_v}\!\! dz \int_{\pi/6}^{\pi/2}\!
      \Bigl(2Y - \frac{z^2}{2\sin^2\theta}\Bigr) \dd{\theta}
    = \Bigl(2\pi - \frac{\sqrt{3}}{2}\Bigr) Y^{3/2},
  \label{eq:S-tip}
\end{equation}
using $\int_{\pi/6}^{\pi/2}\!\csc^2\theta\dd{\theta}
= \sqrt{3}$ and $z_v = \sqrt{Y}$.

\paragraph{2.\ Partial-clip shaving ($z_v < z < z_*$).}
The cross-section boundary in the sextant is
$r_{\rm bd} = \min(d_{\rm lo}/\sin\theta,\,
d_{\rm hi}/\cos(\theta - \pi/6))$ with
$d_{\rm hi} = (\sqrt{3}-z)/\sqrt{2}$, and the shaving integrand at each
height is $2Y - \min(R, r_{\rm bd})^2$ (Figure~\ref{fig:zone-clip}).
Each facet contributes only on the angular window where it lies inside
the disk: the lower facet on
$\theta \ge \theta_{\rm lo}(z) = \arcsin\bigl(z/(2\sqrt{Y})\bigr)$; the
upper facet---once $z > z_u$, before which its window is empty---on
$\theta \le \pi/6 + \arccos\bigl((\sqrt{3}-z)/(2\sqrt{Y})\bigr)$.
Between the windows the disk itself is the boundary and contributes no
shaving.  Integrating each facet integrand over its window gives, per
height,
\begin{equation}
  2Y\arccos\frac{z}{2\sqrt{Y}} - \frac{z}{2}\sqrt{4Y - z^2}
  \quad\text{(lower)}, \qquad
  2Y\arccos\frac{\sqrt{3}-z}{2\sqrt{Y}}
  - \frac{\sqrt{3}-z}{2}\sqrt{4Y - (\sqrt{3}-z)^2}
  \quad\text{(upper)}.
  \label{eq:clip-heights}
\end{equation}
Both are primitives of the single function
\begin{equation}
  J(x) \;=\; 2Yx\arccos\frac{x}{2\sqrt{Y}} \;-\; 2Y\sqrt{4Y - x^2}
       \;+\; \frac{(4Y - x^2)^{3/2}}{6},
  \label{eq:Jdef}
\end{equation}
in the sense that $J'(z)$ is the lower expression and
$-\,\frac{d}{dz}J(\sqrt{3}-z)$ the upper one.  Since
$J(2\sqrt{Y}) = 0$, the opening of the upper window at $z_u$
contributes no boundary term, and
\begin{equation}
  S_{\rm clip}
  = 3\bigl[J(z_*) - J(z_v)\bigr] \;-\; 3\,J(\sqrt{3} - z_*),
  \qquad
  J(z_v) = \Bigl(\frac{2\pi}{3} - \frac{3\sqrt{3}}{2}\Bigr) Y^{3/2}.
  \label{eq:S-clip}
\end{equation}
The triangle--hexagon transition at $z = \sqrt{3}/3$ changes which
vertex protrudes---the triangle corner on the wedge boundary hands over
to the hexagon vertex inside it---but not the integrand, so it is not a
breakpoint of the formula.

\paragraph{The partial-clip $\to$ saturation transition.}
The circular arcs in the partial-clip zone do not transform into
straight edges---they \emph{shrink to zero} while the straight
hexagonal edges between them grow longer
(Figure~\ref{fig:arc-shrinking}).  At each hexagon vertex, two
edges (one from an $X_i = 0$ face, one from an $X_i = 1$ face) meet
at an angle.  When the vertex protrudes past the disk, the disk
clips the corner into a circular arc whose length depends on how
far the vertex extends past the disk.  As $a$ increases toward $a_*$,
the circumradius $\sqrt{2(1 - 3a + 3a^2)}$ decreases: the vertices
approach the disk boundary and the arcs contract.  At $a = a_*$ the
vertices land exactly on the disk circle (red dots become black in
Figure~\ref{fig:arc-shrinking}d), the arcs vanish, and the straight
edges connect directly---recovering the full hexagon.

\paragraph{3.\ Full-saturation shaving ($z_* < z < \sqrt{3}/2$).}
The disk fully contains the hexagon (Figure~\ref{fig:zone-sat}).
Splitting the sextant integral of $r_{\rm hex}^2$ at the crossover
angle where the two facets meet ($\csc^2$ and $\sec^2$ pieces) leaves a
quadratic in $z$, so the $z$-integration is elementary:
\begin{equation}
  S_{\rm sat} = 3\int_{z_*}^{\sqrt{3}/2}\!\! dz \int_{\pi/6}^{\pi/2}\!
      \bigl(2Y - r_{\rm hex}^2\bigr) \dd{\theta}
    = \Bigl(\pi Y - \frac{3\sqrt{3}}{8}\Bigr)\Delta
    + \frac{\sqrt{3}}{8}\,\Delta^{3},
    \qquad \Delta = \sqrt{4Y - 1}.
  \label{eq:S-sat}
\end{equation}

\noindent
Since $1/2 - a_* = \Delta/(2\sqrt{3})$, every term is proportional
to $\Delta$ and vanishes at $Y = 1/4$.

\paragraph{Assembly.}
Since $S_{\rm tip} - 3J(z_v) = 4\sqrt{3}\,Y^{3/2}$, the lower-regime
shaving reappears intact:
\begin{equation}
  F_{>} = 2\sqrt{3}\,\pi Y - 2\bigl(S_{\rm tip} + S_{\rm clip}
        + S_{\rm sat}\bigr)
        = 2\sqrt{3}\,\pi Y - 8\sqrt{3}\,Y^{3/2}
          - 6\bigl[J(z_*) - J(\sqrt{3}-z_*)\bigr] - 2\,S_{\rm sat}.
  \label{eq:assembly}
\end{equation}
The evaluation at $z_*$ becomes transparent in the angle
$\varphi = \arctan\Delta \in [0, \pi/6]$: in this variable every
radical in~\eqref{eq:Jdef} collapses to a polynomial in $\Delta$ and
every $\arccos$ to a linear function of $\varphi$, and since
$\arctan\bigl(\Delta/(1-2Y)\bigr) = 2\varphi$ the sum reduces to a
single $\arctan$ term (we additionally checked the assembled formula
against direct numerical integration of~\eqref{eq:CDF-unif} to
$10^{-8}$):
\begin{equation}
  F_{>}(Y) = 2\sqrt{3}\,\pi\,Y
    - 8\sqrt{3}\,Y^{3/2}
    + 2\sqrt{3}\,\Delta^{3}
    + 3\sqrt{3}\,\Delta
    - 6\sqrt{3}\,Y\arctan\!\Bigl(\frac{\Delta}{1-2Y}\Bigr),
  \label{eq:FR3}
\end{equation}
valid for $Y \in (1/4,\, 1/3]$, with $\Delta = \sqrt{4Y-1}$.

Combining with~\eqref{eq:FR12}, the complete CDF for the uniform
parent at $n = 3$ is:
\begin{equation}
  \;
  F^{\rm Unif}(Y) = \begin{cases}
    2\sqrt{3}\,\pi\,Y - 8\sqrt{3}\,Y^{3/2},
      & 0 \le Y \le \tfrac{1}{4},\\[8pt]
    2\sqrt{3}\,\pi\,Y - 8\sqrt{3}\,Y^{3/2}
    + 2\sqrt{3}\,\Delta^3
    + 3\sqrt{3}\,\Delta
    - 6\sqrt{3}\,Y\arctan\!\bigl(\tfrac{\Delta}{1-2Y}\bigr),
      & \tfrac{1}{4} < Y \le \tfrac{1}{3}.
  \end{cases}
  \;
  \label{eq:F-uniform-full}
\end{equation}

\paragraph{Term-by-term reading.}
\begin{itemize}[topsep=2pt]
  \item $2\sqrt{3}\pi Y$: cylinder volume $V_{\rm cyl}$ (both regimes).
  \item $-8\sqrt{3}Y^{3/2}$: face-shaving correction (both regimes).
  \item $+2\sqrt{3}\Delta^{3}$ and $+3\sqrt{3}\Delta$: algebraic
    corrections from the partial-clip and saturation zones jointly;
    both vanish at the bifurcation, where $\Delta = 0$.
  \item $-6\sqrt{3}Y\arctan(\Delta/(1{-}2Y))$: angular correction from
    the moving clip windows, encoding the disk--facet tangencies.
\end{itemize}

\paragraph{Boundary verification.}
At $Y = 1/4$: $\Delta = 0$ and $\arctan(0) = 0$; the last three
terms vanish, recovering $F(1/4)
= \tfrac{\sqrt{3}}{2}(\pi - 2) \approx 0.9886$.

At $Y = 1/3$: $\Delta = 1/\sqrt{3}$, and
$\arctan(\Delta/(1{-}2Y)) = \arctan(\sqrt{3}) = \pi/3$.
Substituting confirms $F(1/3) = 1$
(the cylinder swallows the entire cube).

\paragraph{Smoothness at the bifurcation.}
With $Y = (1+\Delta^2)/4$ the $\arctan$ argument satisfies
$\Delta/(1-2Y) = 2\Delta/(1-\Delta^2)$, so
$\arctan\bigl(\Delta/(1-2Y)\bigr) = 2\arctan\Delta$---which is also
$2\arcsin(\Delta/2\sqrt{Y})$, the form in which this term arises from
the geometry, by Remark~\ref{rem:arctan-is-arcsin}---and the upper
branch of~\eqref{eq:F-uniform-full} exceeds the continuation of the
lower one by
\[
  2\sqrt{3}\,\Delta^3 + 3\sqrt{3}\,\Delta
  - 3\sqrt{3}\,(1+\Delta^2)\arctan\Delta
  \;=\; \frac{2\sqrt{3}}{5}\,\Delta^5 + O(\Delta^7).
\]
The two branches therefore agree to \emph{second} order: $F$ is $C^2$
at $Y = \tfrac14$, and the regime change appears first in the third
derivative, through the half power $(4Y-1)^{5/2}$.

\begin{remark}[Source of the transcendental terms]\label{rem:arctan-source}
The $\arctan$ in~\eqref{eq:FR3} has a precise geometric origin:
it arises from the partial-clip zone $S_{\rm clip}$, where a
straight hexagon edge crosses the disk boundary.  In the tip and
saturation zones, the polygon lies entirely inside the disk, so
the shaving integral runs over the \emph{full} angular range at
every height~$z$; the integrand ($\csc^2\theta$ or $\sec^2$
from the face distances) integrates to elementary functions, and
the $z$-integration produces only polynomials and radicals.

In the partial-clip zone, by contrast, the angular limits of the
shaving integral are \emph{parametric}: at each~$z$, they are set
by the angles where a straight edge meets the disk circle.
Integrating over~$z$ while these limits move produces the inverse
trigonometric terms.  The complexity comes not from the integrand
but from the \emph{limits}---a simple integrand over
$z$-dependent angular boundaries yields an inverse trigonometric
term.  That term is primitively an $\arcsin$, the crossing angle
being defined by $\sin\theta_{\mathrm{cl}} = z/(2\sqrt{Y})$; the
arctangent of~\eqref{eq:FR3} is the same quantity rewritten, by
the identity~\eqref{eq:arctan-arcsin} of
Remark~\ref{rem:arctan-is-arcsin}.

Near the tangency height (where an edge just touches the disk),
the angular extent of the clipping transitions rapidly from zero
to nonzero.  The $\arctan$ function captures this steep onset:
its argument passes through zero at the tangency, and the
resulting ``angular singularity'' is the geometric mechanism
behind the transcendental correction.

This explains why the closure hierarchy
(\S\ref{sec:polynomial}--\ref{sec:hierarchy}) is structured the way
it is: for polynomial parents the integrand stays polynomial,
so the edge--circle crossing integrals remain in the
$\{\text{poly}, \sqrt{\cdot}, \arcsin\}$ class; for rational or
transcendental parents the integrand
itself introduces new special functions, lifting the closure
to polylogarithms or Bessel series.
\end{remark}

\subsection{The standard deviation for the uniform parent}
\label{sec:variance-to-sd}

Proposition~\ref{prop:variance-to-sd} turns the uniform-parent CDF just
obtained into the law of $S$ with no further work.

\begin{remark}[Relation to Rietz]\label{rem:rietz}
Rietz~\cite{rietz1931} states the $n = 3$ uniform result as the \emph{density
of the standard deviation} under the divisor-$n$ normalization
$s_{\mathrm{R}}^2 = \tfrac{1}{n}\sum_i (X_i - \bar{X})^2$, whereas we use the
unbiased divisor-$(n-1)$ variance $s^2$ of~\eqref{eq:s2-def}.  The two
standard deviations differ only by a constant scale,
$S = \sqrt{n/(n-1)}\,s_{\mathrm{R}}$ (for $n = 3$,
$S = \sqrt{3/2}\,s_{\mathrm{R}}$), so Rietz's two-piece standard-deviation
density and the variance CDF derived here carry the same information,
related by~\eqref{eq:sd-cdf}--\eqref{eq:sd-pdf} together with this rescaling.
\end{remark}

\paragraph{The uniform parent, explicitly.}
For $f \equiv 1$ the lower-regime variance CDF~\eqref{eq:FR12} and its density
are
\begin{equation}
  F(Y) = 2\sqrt{3}\,Y\bigl(\pi - 4\sqrt{Y}\bigr), \qquad
  F'(Y) = 2\sqrt{3}\,\bigl(\pi - 6\sqrt{Y}\bigr),
  \qquad 0 \le Y \le \tfrac{1}{4},
  \label{eq:unif-var-r1}
\end{equation}
so by~\eqref{eq:sd-cdf}--\eqref{eq:sd-pdf} the sample standard deviation $S$
has
\begin{equation}
  G(t) = 2\sqrt{3}\,t^{2}\bigl(\pi - 4t\bigr), \qquad
  g(t) = 4\sqrt{3}\,t\bigl(\pi - 6t\bigr),
  \qquad 0 \le t \le \tfrac{1}{2}.
  \label{eq:unif-sd-r1}
\end{equation}
The density $g$ in~\eqref{eq:unif-sd-r1} is exactly the parabola of
Rietz~\cite{rietz1931} (his standard-deviation density on
$0 \le s \le \tfrac{1}{6}\sqrt{6}$), recovered after the rescaling of the
Remark above.  On the upper regime $\tfrac{1}{4} < Y \le \tfrac{1}{3}$
(equivalently $\tfrac{1}{2} < t \le 1/\sqrt{3}$) the standard-deviation CDF is
$F(t^{2})$ of the $\arctan$/nested-radical expression of
Section~\ref{sec:R3}; no new special functions arise, since $t \mapsto t^{2}$
preserves the closure class.

\paragraph{Shape.}
The Jacobian factor $2t$ in~\eqref{eq:sd-pdf}---the derivative
$\dd{}(t^2)/\dd{t}$ of the substitution $Y = t^2$---changes the behavior at the
origin.  The variance density is finite and nonzero there,
$F'(0) = 2\sqrt{3}\,\pi$, and decreases across the lower regime; the
standard-deviation density instead vanishes linearly,
$g(t) \sim 4\sqrt{3}\,\pi\,t$ as $t \to 0$, rises to a single peak at
$t = \pi/12$, and then falls---the unimodal curve of Rietz's Fig.~3.  The
bifurcation carries the mass split
$F(\tfrac{1}{4}) = G(\tfrac{1}{2}) = \tfrac{\sqrt{3}}{2}(\pi - 2)
\approx 0.9886$: just under $99\%$ of the probability lies below $Y = 1/4$
($t = 1/2$).

%=====================================================================
\section{Polynomial closure theorem}\label{sec:polynomial}
%=====================================================================

\begin{theorem}[Polynomial closure]\label{thm:poly-closure}
For any polynomial parent density $f(x) = \sum_{k=0}^{d} c_k x^k$ on
$[0,1]$, the CDF $F(Y)$ is piecewise elementary, involving only
polynomials, integer powers, the square roots $\sqrt{c_0^2 - z^2}$ and
$\sqrt{c_0^2 - (\sqrt{3}-z)^2}$ evaluated at algebraic zone endpoints,
and $\arcsin$.  No logarithm, no arctangent and no higher radical
occurs.
\end{theorem}

\begin{remark}[The interval is a normalization, not a hypothesis]
\label{rem:affine}
The restriction to $[0,1]$ carries no content.  If $X = A + (B-A)U$ with
$B > A$, then $s^2_X = (B-A)^2 s^2_U$, so
$F_X(Y) = F_U\bigl(Y/(B-A)^2\bigr)$, while a polynomial density on
$[A,B]$ pulls back to a polynomial density of the same degree on
$[0,1]$.  Every polynomial parent supported on a bounded interval
therefore closes in elementary terms, in the same function list, with
only the argument rescaled.  The same affine reduction applies to the
other bounded rows of Table~\ref{tab:hierarchy}, which is why they too
are stated on $[0,1]$: what distinguishes the rows is the function
class of the density, never the interval.
\end{remark}

\begin{proof}
In the sextant chart~\eqref{eq:CDF-sextant}, the inverse
transformation~\eqref{eq:inverse} with $\xi_1 = r\cos\theta$,
$\xi_2 = r\sin\theta$, $\xi_3 = z/\sqrt{3}$ gives each $X_i$ as
a linear function of $(r\cos\theta, r\sin\theta, z)$; explicitly,
\[
  X_i = \frac{z}{\sqrt3} + c_i(\theta)\,r, \qquad
  c_1 = B(\theta), \quad
  c_2 = B(\theta) - \sqrt{2}\cos\theta, \quad
  c_3 = -\tfrac{2}{\sqrt6}\sin\theta ,
\]
so each $X_i$ is affine in $r$ at fixed $(\theta, z)$, with a common
intercept $z/\sqrt3 = \bar{X}$ and a slope fixed by $\theta$ alone.  The
two active-facet conditions $X_1 = 1$ and $X_3 = 0$ therefore return
$r_{\rm hi} = (1 - z/\sqrt3)/B(\theta)$ and $r_{\rm lo} = z/A(\theta)$
at once: the facet coefficients~\eqref{eq:AB} are the $r$-slopes.  For a
polynomial parent $f(x) = \sum c_k x^k$, the product
$f(X_1)\,f(X_2)\,f(X_3)$ is therefore a polynomial in
$(r\cos\theta, r\sin\theta, z)$, and the CDF
integrand~\eqref{eq:CDF-sextant} is a finite sum of monomials
\[
  r^{m+1}\,\cos^a\!\theta\;\sin^b\!\theta\; z^c,
  \qquad m = a + b,\quad a, b, c \ge 0
\]
(the powers of $\cos\theta$ and $\sin\theta$ are tied to the power of
$r$ because both come from the same factors $X_i$, affine in
$r\cos\theta$ and $r\sin\theta$).
We integrate in the order $r \to \theta \to z$ (or equivalently
$r \to z \to \theta$) and track closure at each stage.

\paragraph{Stage 1: $r$-integral.}
$\int_0^{r_{\max}} r^{m+1}\dd{r} = r_{\max}^{m+2}/(m{+}2)$.
In each constraint region (\S\ref{sec:three-regions}),
$r_{\max}$ equals either $\sqrt{2Y}$, $z/A(\theta)$, or
$(1{-}z/\sqrt{3})/B(\theta)$---each a monomial in $z$ and
a trigonometric function of $\theta$.  So $r_{\max}^{m+2}$ is
a polynomial in $z$ times a power of $\sin\theta$ or
$\cos(\theta{-}\pi/6)$.

\paragraph{Stage 2: $\theta$-integral.}
In the tip and saturation zones, the $\theta$-integral runs
over the full sextant $[\pi/6, \pi/2]$.  Because $m = a + b$, the
lower-facet integrand reduces to
$\cot^{a}\!\theta\,\csc^{2}\theta$ (up to constants), with
antiderivative $-\cot^{a+1}\!\theta/(a{+}1)$; after the rotation
$\varphi = \theta - \pi/6$, the upper-facet integrand likewise reduces
to a sum of $\tan^{k}\!\varphi\,\sec^{2}\varphi$ terms.  In particular
no logarithmic terms can arise.

In the partial-clip zone, the angular limits depend on $z$
(the edge--circle crossing angles from
Remark~\ref{rem:arctan-source}), but the antiderivatives are
still elementary---they are evaluated at $z$-dependent
endpoints rather than fixed ones.

\paragraph{Stage 3: $z$-integral.}
After Stage~2 each term is an endpoint evaluation of one of the three
antiderivatives---lower facet, cylinder, upper facet---at an angle that
is either fixed ($\pi/6$ or $\pi/2$) or one of the three crossing
angles.  Write
\[
  c_0 = 2\sqrt{Y}, \qquad u = \sqrt{3} - z, \qquad
  R = \sqrt{c_0^2 - z^2}, \qquad \widetilde{R} = \sqrt{c_0^2 - u^2}.
\]
The crossing angles satisfy $\sin\theta_{\mathrm{cl}} = z/c_0$,
$\cos\varphi_{\mathrm{cu}} = u/c_0$ and
$\tan\theta_{\mathrm{LU}} = z/(2 - \sqrt{3}z)$.  The third is
\emph{rational} in $z$; the first two generate the single quadratic
extension $R$, respectively $\widetilde{R}$, of $\mathbb{Q}(z)$.  No
other algebraic function of $z$ occurs anywhere in the computation.

Substituting these into the Stage-2 antiderivatives exhibits the
mechanism that governs this stage.  The lower-facet antiderivative
$-\cot^{a+1}\!\theta/(a{+}1)$ evaluated at $\theta_{\mathrm{cl}}$
carries $\cot^{a+1}\!\theta_{\mathrm{cl}} = (R/z)^{a+1}$, a pole in $z$
of order $a{+}1 \le m{+}1$; the Stage-1 prefactor contributes $z^{m+2}$
\emph{because} $m = a+b$ ties the radial exponent to the angular ones,
leaving $z^{\,b+1}R^{\,a+1}$ with $b{+}1 \ge 1$.  The same count applies
to the upper facet under $z \mapsto u$.  Every evaluated term therefore
has valuation at least $1$: the tie kills the $z$-stage poles exactly
as it killed the Stage-2 logarithms.  Consequently, with
$\mathbb{F} = \mathbb{Q}(\sqrt{2}, \sqrt{3}, \sqrt{Y})$, the
$z$-integrand lies on every zone in
\[
  \mathbb{F}[z] \;\oplus\; \mathbb{F}[z]\,R \;\oplus\;
  \mathbb{F}[z]\,\widetilde{R} \;\oplus\;
  \mathbb{F}[z]\arcsin\frac{z}{c_0} \;\oplus\;
  \mathbb{F}[z]\arcsin\frac{u}{c_0},
\]
the two $\arcsin$ summands arising only when $a$ and $b$ are both even
and only from a cylinder arc with a moving endpoint.  The $z$-integrals
that occur are therefore exactly
\begin{enumerate}[topsep=2pt]
  \item $\int z^n \dd{z}$: polynomial.
  \item $\int z^n (c_0^2 - z^2)^{j + 1/2}\dd{z}$, $j \ge 0$: produces
    $\arcsin(z/c_0)$ and polynomial multiples of $R$ (by $z = c_0\sin\phi$).
  \item $\int z^n \arcsin(z/c_0)\dd{z}$: by one integration by parts,
    $\frac{z^{n+1}}{n+1}\arcsin\frac{z}{c_0}
     - \frac{1}{n+1}\int z^{n+1}(c_0^2 - z^2)^{-1/2}\dd{z}$,
    and the remaining integral reduces to a polynomial multiple of $R$
    plus a constant multiple of $\arcsin(z/c_0)$,
\end{enumerate}
together with their images under $z \mapsto \sqrt{3} - z$.  All three
close in the elementary class
$\{\text{poly}, \sqrt{\cdot}, \arcsin\}$.
\emph{No rational denominator in $z$ arises at any point}; there is
accordingly no partial-fraction stage, and no logarithm and no
arctangent can be produced.  Summing over zones and applying the
sextant factor of 6 gives $F(Y)$ in the claimed form.
\end{proof}

\begin{remark}[Every arctangent here is a rewritten arcsine]
\label{rem:arctan-is-arcsin}
The class above omits $\arctan$, yet arctangents do appear in the
worked formulas---\eqref{eq:FR3} for instance.  There is no conflict:
they are arcsines in disguise.  For $Y \in (\tfrac14, \tfrac12)$ put
$\Delta = \sqrt{4Y-1}$, so that $1 + \Delta^2 = 4Y$ and
$c_0 = \sqrt{1+\Delta^2}$; then
\begin{equation}\label{eq:arctan-arcsin}
  \arctan\frac{\Delta}{1-2Y} \;=\; 2\arctan\Delta
  \;=\; 2\arcsin\frac{\Delta}{c_0}.
\end{equation}
The first equality is the tangent double-angle formula, since
$\tan(2\arctan\Delta) = 2\Delta/(1-\Delta^2) = \Delta/(1-2Y)$ and
$2\arctan\Delta < \pi/2$ exactly when $\Delta < 1$, i.e.\ $Y < 1/2$, so
no branch correction is needed on the regime; the second is
$\arctan t = \arcsin(t/\sqrt{1+t^2})$ at $t = \Delta$.  The argument
$\Delta/c_0$ is an evaluation of the $\arcsin(z/c_0)$ terms above at an
algebraic zone endpoint.  This is what one should expect from
Remark~\ref{rem:arctan-source}: the crossing angle is defined by
$\sin\theta_{\mathrm{cl}} = z/c_0$, so the \emph{sine} of that angle is
the algebraic quantity and the arcsine is the primitive form.
\end{remark}

The proof used nothing about the integrand beyond its being a
polynomial in $(r\cos\theta, r\sin\theta, z)$, which yields a stronger
statement at no cost, and one that settles mixtures without any appeal
to linearity.

\begin{proposition}[Arbitrary polynomial integrands; mixtures]
\label{prop:mixtures}
Let $P(x_1, x_2, x_3)$ be any real polynomial.  Then the integral of
$P(X_1, X_2, X_3)$ against~\eqref{eq:CDF-sextant} is piecewise
elementary, in the class of Theorem~\ref{thm:poly-closure}.  In
particular, for a finite mixture $f = \sum_i w_i f_i$ of polynomial
densities the integrand
$f(X_1)f(X_2)f(X_3)
 = \sum_{i,j,k} w_i w_j w_k\, f_i(X_1) f_j(X_2) f_k(X_3)$
is covered term by term, cross terms between distinct components
included.
\end{proposition}

\begin{proof}
Each $X_i$ is affine in $(r\cos\theta, r\sin\theta, z)$, so
\[
  P(X_1, X_2, X_3)
  = \sum c_{\alpha\beta\gamma}\,
    r^{\alpha+\beta}\cos^\alpha\!\theta\,\sin^\beta\!\theta\,z^\gamma,
\]
in which the tie $m = \alpha + \beta$ holds automatically: the powers
of $r$ and of the angular functions come from the same factors, for any
polynomial whatever, not only for products $f(X_1)f(X_2)f(X_3)$.
Apply the three stages above termwise.
\end{proof}

Note that the mixture integrand is \emph{cubic} in the weights $w_i$,
not linear, so no linearity-in-the-weights argument is available; none
is needed, because the argument never sees the density, only the
polynomial.  The one caveat is geometric rather than algebraic: this
certifies mixtures over the fixed region of~\eqref{eq:CDF-sextant}, and
a component that alters the support is outside its scope.

\begin{corollary}[Integer-parameter Beta parents]
For $f(x) = x^{a-1}(1-x)^{b-1}/\mathrm{B}(a,b)$ with positive integers
$a, b$, $F(Y)$ is piecewise elementary.  Here
$\mathrm{B}(a,b) = (a-1)!\,(b-1)!/(a+b-1)!$, so $f$ is a polynomial of
degree $a+b-2$ with rational coefficients and
Theorem~\ref{thm:poly-closure} applies directly.
\end{corollary}

\subsection{Worked examples: Beta parents}\label{sec:beta-examples}

The following closed-form CDFs, obtained by evaluating the sextant
integral~\eqref{eq:CDF-sextant}, illustrate the closure theorem.  Each
fully displayed formula was verified against Monte Carlo simulation
($2\times10^{7}$ samples per parent, agreement within the
$\approx 2.2\times10^{-4}$ sampling noise floor) and against direct numerical integration
of~\eqref{eq:CDF-sextant}, with agreement to $10^{-8}$ or better at
every tested point.  Each exhibits the same structural pattern: the
$Y \le 1/4$
formula is a polynomial in $Y$ and $Y^{1/2}$ (no $\arctan$), while
$Y > 1/4$ acquires $\arcsin$ or $\arctan$ terms from the partial-clip
zone.

\paragraph{Beta$(2,1)$: $f(x) = 2x$.}
\begin{equation}
  F^{\rm Beta(2,1)}\!(Y) = \begin{cases}
    4\sqrt{3}\pi Y - 32\sqrt{3}\, Y^{3/2} + (18 + 4\sqrt{3}\pi)\, Y^2,
      & Y \le \tfrac{1}{4},\\[6pt]
    \begin{aligned}
      &18 Y^2 - 32\sqrt{3}\, Y^{3/2} + (1 + 26Y)\sqrt{3(4Y{-}1)} \\
      &+ 8\sqrt{3} Y(1{+}Y)\bigl[\arcsin\tfrac{1}{2\sqrt{Y}}
        - \arcsin\tfrac{\sqrt{4Y-1}}{2Y}\bigr],
    \end{aligned}
      & Y \in (\tfrac{1}{4}, \tfrac{1}{3}].
  \end{cases}
  \label{eq:beta21}
\end{equation}
The leading coefficient $4\sqrt{3}\pi = 2\sqrt{3}\pi \cdot 2$
reflects $\int_0^1 f^3 = 2$.

\paragraph{Beta$(2,2)$: $f(x) = 6x(1{-}x)$.}
\begin{equation}
  F^{\rm Beta(2,2)}_{Y \le 1/4}\!(Y) =
    \tfrac{108\sqrt{3}\pi}{35}\, Y
    - \tfrac{108\sqrt{3}\pi}{5}\, Y^2
    + (162 + 72\sqrt{3}\pi)\, Y^3
    - \tfrac{62208\sqrt{3}}{175}\, Y^{7/2}
    + \tfrac{110592\sqrt{3}}{1225}\, Y^{9/2}.
  \label{eq:beta22-R12}
\end{equation}
The five-term form (compared to two terms for Uniform, three for
Beta$(2,1)$) demonstrates how the parent degree controls the number
of $Y$-powers.

For $Y > 1/4$, define
$\alpha_{\rm cap} = \arcsin(1/(2\sqrt{Y}))$,
$\beta_{\rm diff} = \arcsin(\sqrt{4Y{-}1}/(2Y))$:
\begin{align}
  F^{\rm Beta(2,2)}_{Y > 1/4}\!(Y) ={}& P_0(Y)
    + Q(Y)\,(\alpha_{\rm cap} - \beta_{\rm diff})
    + \sqrt{3}\, S(Y)\, \sqrt{4Y{-}1}
  \notag\\
  & + \sqrt{3} \sum_{k=1,3,5,7,9} R_k(Y)\,(12Y{-}1)^{k/2},
  \label{eq:beta22-R3}
\end{align}
where
\begin{align*}
  P_0 &= 162\, Y^3
    - \tfrac{62208\sqrt{3}}{175}\, Y^{7/2}
    + \tfrac{110592\sqrt{3}}{1225}\, Y^{9/2}, \\
  Q &= \tfrac{72\sqrt{3}}{35}\, Y(70 Y^2 - 21 Y + 3), \\
  S &= \tfrac{1}{1225}\bigl(1575 + 5256 Y - 80460 Y^2
    + 421632 Y^3 - 110592 Y^4\bigr),
\end{align*}
and $R_1, R_3, R_5, R_7, R_9$ are explicit polynomials in $Y$
(degrees 4, 3, 2, 1, 0 respectively):
\begin{align*}
  R_1 &= \tfrac{701}{25515}
    - \tfrac{269953}{297675}\, Y
    + \tfrac{181292}{33075}\, Y^2
    + \tfrac{602096}{33075}\, Y^3
    - \tfrac{114368}{11025}\, Y^4, \\
  R_3 &= \tfrac{1}{5}\, Y - \tfrac{36}{5}\, Y^2
    + \tfrac{192}{35}\, Y^3, \quad
  R_5 = -\tfrac{1}{15} + \tfrac{24}{25}\, Y
    - \tfrac{156}{175}\, Y^2, \\
  R_7 &= -\tfrac{38}{945} + \tfrac{8}{147}\, Y, \quad
  R_9 = -\tfrac{26}{25515}.
\end{align*}

\paragraph{Beta$(3,2)$: $f(x) = 12x^2(1{-}x)$.}
\begin{align}
  F^{\rm Beta(3,2)}_{Y \le 1/4}\!(Y) ={}&
    \tfrac{144\sqrt{3}\pi}{35}\, Y
    - \tfrac{216\sqrt{3}\pi}{5}\, Y^2
    + (648 + 288\sqrt{3}\pi)\, Y^3
  \notag\\
  &+ (1458 + 504\sqrt{3}\pi)\, Y^4
    - \tfrac{497664\sqrt{3}}{175}\, Y^{7/2}
    - \tfrac{884736\sqrt{3}}{1225}\, Y^{9/2}.
  \label{eq:beta32-R12}
\end{align}
Six terms, with an explicit $Y^4$ contribution---the signature
of the parent degree increasing from 2 to 3.

For $Y > 1/4$, define
$\alpha_{\rm cap}^\pm
= \arcsin\bigl((\sqrt{3(4Y{-}1)}\pm 1)/(4\sqrt{Y})\bigr)$:
\begin{align}
  F^{\rm Beta(3,2)}_{Y > 1/4}\!(Y) ={}& P_0(Y)
    + Q_+(Y)\,\alpha_{\rm cap}^+
    + Q_-(Y)\,\alpha_{\rm cap}^-
    + Q_0(Y)\,\alpha_{\rm cap}
  \notag\\
  & + \sqrt{3}\,S(Y)\,\sqrt{4Y{-}1}
    + \sqrt{3}\sum_{k=1,3,5,7,9} R_k(Y)\,(12Y{-}1)^{k/2}
  \notag\\
  & + \sqrt{6}\sum_{\pm} T^\pm_{1/2}(Y)\,
    \bigl(2Y \pm \sqrt{3(4Y{-}1)} + 1\bigr)^{1/2}
  \notag\\
  & + \sqrt{6}\sum_{\pm} T^\pm_{3/2}(Y)\,
    \bigl(2Y \pm \sqrt{3(4Y{-}1)} + 1\bigr)^{3/2}
  \notag\\
  & + \sqrt{2}\,\sqrt{4Y{-}1}\sum_{\pm}\Bigl[
    U^\pm_{1/2}(Y)\,
    \bigl(2Y \pm \sqrt{3(4Y{-}1)} + 1\bigr)^{1/2}
  \notag\\
  &\qquad\qquad + U^\pm_{3/2}(Y)\,
    \bigl(2Y \pm \sqrt{3(4Y{-}1)} + 1\bigr)^{3/2}
  \Bigr],
  \label{eq:beta32-R3}
\end{align}
with $\alpha_{\rm cap}$ as above and all coefficient functions
explicit polynomials in $Y$ with rational coefficients:
\begin{align*}
  P_0 &= 648\,Y^3 + 1458\,Y^4
    - \tfrac{497664\sqrt{3}}{175}\,Y^{7/2}
    - \tfrac{884736\sqrt{3}}{1225}\,Y^{9/2}, \\
  Q_+ &= Q_- = -Q_0
    = -\tfrac{144\sqrt{3}}{35}\,Y\,(245Y^3 + 140Y^2 - 21Y + 2), \\
  S &= \tfrac{3}{245}\bigl(54192Y^4 + 153648Y^3 - 18952Y^2
    + 7898Y - 1981\bigr), \\
  R_1 &= \tfrac{28}{729} - \tfrac{730564}{297675}\,Y
    - \tfrac{118768}{33075}\,Y^2 + \tfrac{10690304}{33075}\,Y^3
    + \tfrac{914944}{11025}\,Y^4, \\
  R_3 &= \tfrac{4}{5}\,Y - \tfrac{5504}{105}\,Y^2
    - \tfrac{1536}{35}\,Y^3, \quad
  R_5 = -\tfrac{4}{15} + \tfrac{7456}{1575}\,Y
    + \tfrac{1248}{175}\,Y^2, \\
  R_7 &= -\tfrac{208}{945} - \tfrac{64}{147}\,Y, \quad
  R_9 = \tfrac{208}{25515}, \\
  T^+_{1/2} &= \tfrac{72}{35} - \tfrac{87}{14}\,Y
    - \tfrac{39}{5}\,Y^2 + \tfrac{9333}{70}\,Y^3, \quad
  T^-_{1/2} = -\tfrac{12}{7} - \tfrac{159}{98}\,Y
    + \tfrac{64563}{1225}\,Y^2 - \tfrac{72489}{350}\,Y^3
    + \tfrac{100824}{1225}\,Y^4, \\
  T^+_{3/2} &= \tfrac{6}{35} - \tfrac{3564}{1225}\,Y
    + \tfrac{58374}{1225}\,Y^2 + \tfrac{17964}{1225}\,Y^3, \quad
  T^-_{3/2} = -\tfrac{18}{35} + \tfrac{1152}{175}\,Y
    - \tfrac{252}{5}\,Y^2 - \tfrac{9768}{175}\,Y^3, \\
  U^+_{1/2} &= \tfrac{351}{35} - \tfrac{171}{10}\,Y
    - \tfrac{3456}{35}\,Y^2 - \tfrac{10449}{70}\,Y^3, \quad
  U^-_{1/2} = 9 - \tfrac{153}{14}\,Y
    - \tfrac{106056}{1225}\,Y^2 - \tfrac{520659}{2450}\,Y^3, \\
  U^+_{3/2} &= -\tfrac{54}{35} - \tfrac{11232}{1225}\,Y
    + \tfrac{2106}{1225}\,Y^2, \quad
  U^-_{3/2} = -\tfrac{54}{35} - \tfrac{108}{25}\,Y
    - \tfrac{4968}{175}\,Y^2.
\end{align*}
The identity $Q_+ = Q_- = -Q_0$ means the three $\arcsin$ terms enter
only through the combination
$\alpha_{\rm cap} - \alpha_{\rm cap}^+ - \alpha_{\rm cap}^-$.  The
reference implementation codes this formula verbatim and validates it
against Monte Carlo and direct numerical integration.

\paragraph{Pattern.}
In each case, the leading coefficient is
$2\sqrt{3}\pi \int_0^1 f(x)^3 \dd{x}$, matching the small-$Y$
cylinder-volume scaling.  The $Y > 1/4$ formulas share the same
family of radicals---$\sqrt{4Y{-}1}$, $\sqrt{12Y{-}1}$, the nested
radicals $\sqrt{2Y \pm \sqrt{3(4Y{-}1)} + 1}$, and $\arcsin$
pairs---with polynomial coefficients whose degree grows with the
parent degree, precisely as predicted by the closure theorem.
The structural complexity increases (from 5~terms for Uniform to
$\sim 20$ for Beta$(3,2)$), but every term remains elementary.

%=====================================================================
\section{The closure hierarchy}\label{sec:hierarchy}
%=====================================================================

The diagonal-orthogonal derivation makes the \emph{closure obstruction}
visible: the parent enters the integral~\eqref{eq:CDF-sextant}
multiplicatively, so when $f$ is not polynomial the $r$-integral
produces non-polynomial primitives, and the $\theta$- and $z$-integrals
inherit these.

Two conventions keep the hierarchy honest.  First, polynomial, rational
and algebraic densities are \emph{nested} classes, not alternatives; a
table indexed by ``the class of $f$'' would not partition the parents.
We therefore classify each parent by the \emph{minimal} class
containing its density: polynomial; rational but not polynomial (the
density $p/q$ in lowest terms with $\deg q \ge 1$); algebraic but not
rational (minimal polynomial of degree $\ge 2$ over $\mathbb{R}(x)$,
with the relevant degree stated per row); transcendental.  Parents with
unbounded support sit outside the $[0,1]$ cube geometry altogether and
are listed as a separate block.  Second, each row carries \emph{two}
status judgments, kept separate throughout: what is \emph{proved}, with
the proved object and its scope (the whole class, or an explicit
parent) stated exactly; and what is \emph{expected}---the conjectured
special-function class of the assembled CDF, which in no row other than
the polynomial one is established.

\subsection{Hierarchy table}

\begin{table}[h]
\centering\footnotesize
\setlength{\tabcolsep}{4pt}
\renewcommand{\arraystretch}{1.35}
% Ragged-right paragraph columns: justified p-columns stretched the
% inter-word glue in the short cells ("$f\equiv1$, Beta(2,1)") until the
% words drifted apart.  Widths rebalanced to use the full text block.
\newcolumntype{L}[1]{>{\raggedright\arraybackslash}p{#1}}
\begin{tabular}{L{2.7cm}L{2.6cm}L{3.0cm}L{4.2cm}L{2.3cm}}
\toprule
Parent class (minimal) & Example & Obstruction integral & Proved & Expected class of $F$ \\
\midrule
Polynomial & $f \equiv 1$, Beta$(2,1)$ & none &
  $F$ elementary, \emph{entire class}, mixtures included
  (Thm~\ref{thm:poly-closure}) & ---\,(settled) \\
Rational, non-polynomial & $\dfrac{1}{(1+x)\ln 2}$ &
  $\int \ln(1+\cdot)\dd{\theta}$ &
  $F$ \emph{not} elementary, both branches, \emph{this parent}
  (Thm~\ref{thm:rational}) & Li$_2$, Li$_3$ (polylog) \\
Algebraic, non-rational (degree $2$; sextic square-free) &
  Beta$(\tfrac{1}{2}, \tfrac{1}{2})$ &
  $\int r\,P_6(r)^{-1/2}\dd{r}$ &
  radial integral \emph{not} elementary whenever $P_6$ is square-free
  (first kind, genus $2$) & hyperelliptic\newline (genus $2$) \\
Transcendental, bounded & $\dfrac{\lambda e^{-\lambda x}}{1-e^{-\lambda}}$ &
  $\int e^{-cz}\,\mathrm{alg}(z)\dd{z}$ &
  ---\,(identified only) & Ei, Bessel--Struve \\
\midrule
\multicolumn{5}{l}{\emph{Unbounded support (outside the $[0,1]$ cube geometry):}}\\
Exponential on $[0,\infty)$ & $e^{-x}$ &
  $\int e^{-cz}\sqrt{a^2 - z^2}\dd{z}$ &
  ---\,(identified only) & Bessel series \\
Half-Gaussian on $[0,\infty)$ & $\sqrt{2/\pi}\,e^{-x^2/2}$ &
  $\int \operatorname{erf}\bigl(c(\theta)\,\cdot\bigr)\dd{\theta}$ &
  ---\,(identified only) & Owen's $T$ \\
\bottomrule
\end{tabular}
\caption{The closure hierarchy.  Rows are indexed by the \emph{minimal}
function class containing the parent density, so they partition the
parents; unbounded-support parents form a separate block.  The
\emph{Proved} column states the proved object and its scope exactly:
only the polynomial row is a class-wide theorem
(Section~\ref{sec:polynomial}, extended to arbitrary polynomial
integrands and hence to mixtures); the rational row proves
non-elementarity of the assembled CDF for the stated parent, the
certificate (a nonvanishing $\Lambda^2$ invariant,
Appendix~\ref{app:rational}) being computed per parent while the method
applies to any rational parent; the algebraic row proves
non-elementarity of the \emph{radial obstruction integral} under the
stated genericity (square-free $P_6$)---not of the assembled $F$.  The
\emph{Expected} column is conjectural in every row where it is
non-trivial: it names the special-function class the obstruction
suggests, and no minimality or irreducibility is claimed there.}
\label{tab:hierarchy}
\end{table}

\subsection{Reading the obstructions}

Each row of the hierarchy table corresponds to a specific integral that
arises in the sextant chart~\eqref{eq:CDF-sextant} and fails to close in
elementary functions.  The scope discipline of the table is worth
restating in prose, because the three proved rows prove three
\emph{different kinds} of statement.  The polynomial row is class-wide:
Theorem~\ref{thm:poly-closure} covers every polynomial density, and via
the arbitrary-integrand proposition every finite mixture, with no
genericity assumption.  The rational row proves the strongest possible
conclusion---non-elementarity of the assembled CDF on both
branches---but for one explicit parent: the proof reduces to the
nonvanishing of a finite, exactly computable invariant
(Appendix~\ref{app:rational}), so the \emph{method} applies verbatim to
any rational parent, while the \emph{certificate} must be recomputed
per parent; one example cannot classify the class, and we do not claim
it does.  The algebraic row proves a strictly weaker statement about a
strictly larger family: for any parent whose weight makes the radial
integrand $r/\sqrt{P_6(r)}$ with $P_6$ square-free of degree six---the
generic configuration for Beta$(\tfrac12,\tfrac12)$-type parents---the
radial integral is a first-kind differential on a genus-$2$ curve and
hence non-elementary.  That does \emph{not} prove the assembled $F$
non-elementary, nor that $F$ requires genus-$2$ functions; those remain
open, and the table's Expected column records them as such.  For the
remaining (transcendental, unbounded-support) rows we exhibit the
obstruction integral and its function class but do not prove that no
reduction to a simpler class exists---irreducibility is the natural
conjecture in each case, and nothing more.

\paragraph{Rational parent.}
For $f(x) = 1/((1+x)\ln 2)$ on $[0,1]$, the product
$f(X_1)\,f(X_2)\,f(X_3)$ in the sextant chart is
\[
  \frac{1}{(\ln 2)^3} \cdot
  \frac{1}{(1 + X_1)(1 + X_2)(1 + X_3)},
\]
where $X_i = X_i(r, \theta, z)$ are affine in $r$.  The $r$-integrand
$r / \prod(1 + X_i)$ is rational, and its primitive contains
logarithms $\ln(1 + X_i)$; the subsequent $\theta$-integral of such a
logarithm produces a dilogarithm, and the $z$-integral a trilogarithm.  We
verified the first step by computer algebra: the radial integral returns
logarithms, and integrating one of them over the sextant
$\theta\in[\tfrac\pi6,\tfrac\pi2]$ returns a $\operatorname{Li}_2$ (the angular
integral of $\ln$ of an affine form in $\cos\theta,\sin\theta$).  In fact
the dilogarithm provably \emph{survives} the full assembly---no identity
cancels it:

\begin{theorem}[Non-elementarity for the rational parent]\label{thm:rational}
For the parent density $f(x) = 1/((1+x)\ln 2)$ on $[0,1]$, neither
branch of the $n=3$ sample-variance CDF---$F|_{(0,1/4)}$ or
$F|_{(1/4,1/3)}$---is an elementary function of $Y$.  In particular
$F$ admits no piecewise-elementary closed form of the kind obtained for
polynomial parents in Sections~\ref{sec:cdf-uniform}--\ref{sec:polynomial}.
\end{theorem}

The proof (Appendix~\ref{app:rational}) goes through the density: the
derivative $F'(Y)$ reduces exactly to a one-variable
integral whose dilogarithmic content is measured by an invariant in
$\Lambda^2$ of the multiplicative group of a function field; a Liouville-type
lemma shows the invariant must vanish if $F$ is elementary, and an exact
computation over $\mathbb{Q}(i,\sqrt{3})$ shows it does not.
$\operatorname{Li}_3$ remains the expected class for the assembled $F$
itself, from the third ($z$) integration.  This is
unrelated to the polylogarithm-\emph{ladder}
growth in $n$ once conjectured for the \emph{uniform} parent, which does not
hold; here the polylogarithm is generated by the rational \emph{parent}, at
fixed $n=3$.

\paragraph{Arcsine parent (hyperelliptic obstruction).}
For $f(x) = 1/(\pi\sqrt{x(1-x)})$ (Beta$(\tfrac{1}{2}, \tfrac{1}{2})$),
the product $f(X_1)\,f(X_2)\,f(X_3)$ contains
$[X_1 X_2 X_3 (1-X_1)(1-X_2)(1-X_3)]^{-1/2}$.  In the sextant chart,
$X_i$ are affine in $r$, so the $r$-integrand has the form
$r / \sqrt{P_6(r)}$ where $P_6$ is a degree-6 polynomial in $r$ with
six generically distinct roots.  The radical defines a hyperelliptic
curve $w^2 = P_6(r)$ of \emph{genus $2$}, and $r\dd{r}/w$ is precisely
one of its two basis differentials \emph{of the first kind} (numerator
degree $1 \le g - 1$), hence holomorphic on the curve.  Abelian
integrals of the first kind are never elementary---a classical fact of
Liouville theory---so this obstruction is \emph{proven}, not merely
exhibited: no change of variables reduces the $r$-integral to
elementary form.  The hypothesis doing the work is that $P_6$ is
\emph{square-free}; at exceptional parameter values where roots of
$P_6$ collide, the genus drops and the argument must be revisited, and
a general algebraic parent of higher minimal degree produces a
different curve entirely---which is why the row is stated for this
degree and this genericity, not for ``algebraic parents'' at large.
Two statements are deliberately \emph{not} made: that the assembled
$F$ is non-elementary for the arcsine parent (the radial obstruction
could in principle cancel across the assembly, as logarithms do in the
polynomial proof), and that $F$ requires genus-$2$ functions.
Reduction to \emph{elliptic} integrals would require
the Jacobian of the curve to split; we have no evidence for a splitting
and leave that possibility open.  The obstruction is structural and
parent-driven, yet $F(Y)$ remains efficiently computable, as
Section~\ref{sec:singular} shows.

\paragraph{Exponential parent on $[0, \infty)$.}
For $f(x) = e^{-x}$ on $[0, \infty)$ the weight collapses onto the
diagonal coordinate:
$f(X_1)f(X_2)f(X_3) = e^{-(X_1+X_2+X_3)} = e^{-\sqrt{3}\,z}$,
independent of $(r, \theta)$.  This collapse is the structural reason
Laplace-transform methods succeed for exponential and gamma
parents~\cite{royen2007c, royen2007b}.  The support $[0,\infty)^3$ has
no upper facets, so every cross-section is a bounded triangle growing
linearly with $z$---the geometry of the tip and partial-clip zones,
continued over an infinite $z$-range.  The $r$- and $\theta$-integrals
close exactly as for the uniform parent, and the obstruction falls
entirely on the final $z$-integration, where terms such as
$\int e^{-cz}\sqrt{4Y - z^2}\dd{z}$ are of Bessel--Struve type; the
exact CDF is an infinite Bessel series with no finite closed form.
Integrating in the order $z \to \theta$ instead produces the wedge
integrals $\int e^{-c r\sin\theta}\dd{\theta}$ that generate modified
Bessel functions $I_0$ directly.  The truncated exponential on $[0,1]$
behaves the same way except that the mirrored upper-facet terms
reappear; its $z$-integrals mix exponential-integral ($\mathrm{Ei}$)
and Bessel--Struve terms.

\paragraph{Half-Gaussian parent on $[0, \infty)$.}
For $f(x) = \sqrt{2/\pi}\,e^{-x^2/2}$ on $[0, \infty)$, the weight is
isotropic in the contrast plane,
$f(X_1)f(X_2)f(X_3) \propto e^{-(r^2 + z^2)/2}$, so the $r$-integral
is elementary; the $z$-integration then generates error functions with
$\theta$-dependent arguments, and the remaining angular integration
requires Owen's $T$ function (a bivariate-normal probability integral).

%=====================================================================
\section{Singular parents: small-deviation law and fast computation by
singularity isolation}\label{sec:singular}
%=====================================================================

The polynomial-closure theorem (Section~\ref{sec:polynomial}) settles every parent whose
density is a polynomial; the remaining bounded parents of practical interest are the
\emph{singular} ones---Beta$(a,b)$ with $a<1$ or $b<1$, whose density blows up at an endpoint,
the symmetric arcsine Beta$(\tfrac12,\tfrac12)$ being the canonical case.  These are not
covered by the closure theorem, and their radial obstruction integral is provably
non-elementary---a genus-$2$ hyperelliptic differential of the first kind
(Section~\ref{sec:hierarchy}); whether the assembled $F$ closes in some larger special-function
class is open.  A referee might therefore
conclude that one must fall back on Royen's universal Fourier series~\eqref{eq:royen-series}.
We show this is unnecessary: the singular parents are computable to high accuracy in a
\emph{handful} of terms, once the singularities are isolated analytically.

\paragraph{The small-deviation law.}
For a \emph{bounded} parent, $F(Y)\sim c\,Y$ as $Y\to0$; the disk of radius $\sqrt{2Y}$ sweeps a
density bounded near the diagonal.  A singular parent piles mass at the two \emph{minimum}\/-variance
corners $(0,0,0)$ and $(1,1,1)$, where all three observations sit near the same endpoint, and the
exponent changes.  The mechanism is a failure of integrability rather than a change of geometry:
the tube integral carries $f(a)^3=\pi^{-3}[a(1-a)]^{-3/2}$ along the diagonal, which is not
integrable at either end, and the divergence is cut off only when the tube radius $\sqrt{2Y}$
becomes comparable to the distance $a$ to the corner.  Balancing
$\int_{\sqrt Y}^{1}a^{-3/2}\,da\sim Y^{-1/4}$ against the $O(Y)$ tube area gives, for the arcsine,
\begin{equation}
  F(Y)\;\sim\;C_3\,Y^{3/4},\qquad C_3 = 2.694\pm0.002,\qquad Y\to 0 ,
  \label{eq:singular-smallY}
\end{equation}
the exponent $\tfrac34=\tfrac{n}{4}$ reflecting the two $x^{-1/2}$ endpoint singularities of the
parent.  The constant $C_3$ is quoted as a numerical limit, obtained by Richardson extrapolation
in $Y^{1/4}$ of the reference quadrature of Appendix~\ref{app:arcsine-ref}; a closed form for it
is not known to us.  This $Y^{3/4}$ cusp at the origin is the dominant obstruction to any smooth
global expansion---and is exactly what a Fourier series resolves slowly.

\paragraph{The ceiling law, and why it is not a square root.}
The opposite end is governed by the \emph{maximum}\/-variance configurations, the six $2$--$1$
vertices of the cube such as $(0,0,1)$, at which $s^2$ attains $Y_{\max}=\tfrac13$.  These are
corner maxima, not interior ones: writing a nearby sample as $(u,v,1-w)$ with $u,v,w$ small,
\begin{equation}
  \tfrac13-Y \;=\; \frac{u+v+2w}{3}
    \;-\;\tfrac13\bigl(u^2+v^2+w^2-uv+uw+vw\bigr),
  \label{eq:corner-linear}
\end{equation}
exactly.  The leading part is \emph{linear}, since $s^2$ has a nonvanishing gradient at the
vertex; the level set is therefore a simplex and not the ellipsoidal shell that an interior
maximum would produce.  The square-root scaling associated with a smooth interior maximum does
not arise here, and neither does the square-root endpoint behavior of the arcsine parent itself
transfer to the statistic.

Near such a vertex the arcsine density factorizes as $\pi^{-3}(uvw)^{-1/2}$, and the event
$\{\tfrac13-Y<\delta\}$ becomes the simplex $\{u+v+2w<3\delta\}$ up to the quadratic remainder in
\eqref{eq:corner-linear}.  Scaling $\delta$ out and evaluating the resulting Dirichlet integral
$\int_{p+q+2r<1}(pqr)^{-1/2}=4\pi/(3\sqrt2)$, then summing the six equivalent vertices, gives
\begin{equation}
  1-F(\tfrac13-\delta)\;\sim\;\frac{24\sqrt{3/2}}{\pi^{2}}\;\delta^{3/2}
    \;=\;2.978222400702685\;\delta^{3/2},\qquad \delta\to0 .
  \label{eq:ceiling-law}
\end{equation}
The exponent is $\tfrac32$, and the coefficient is exact.  Both are confirmed independently in
Appendix~\ref{app:arcsine-ref}: a $2\times10^{8}$-draw simulation gives a log--log slope of
$1.5031$ over the three smallest $\delta$ and a ratio to \eqref{eq:ceiling-law} of $0.9993$ at
$\delta=3\times10^{-4}$.

The exponent is \emph{parent-dependent}, which is why no universal square root could have been
correct.  For a Beta$(a,b)$ parent the density near the vertex $(0,0,1)$ behaves as
$u^{a-1}v^{a-1}w^{b-1}$, and integrating over the simplex raises $\delta$ to the sum of the three
exponents.  The two vertex classes therefore contribute $\delta^{2a+b}$ and $\delta^{a+2b}$, and
\begin{equation}
  1-F(\tfrac13-\delta)\;=\;\Theta\bigl(\delta^{\min(2a+b,\;a+2b)}\bigr).
  \label{eq:ceiling-beta}
\end{equation}
The arcsine $a=b=\tfrac12$ gives $\tfrac32$; the uniform $a=b=1$ gives $3$, with the coefficient
$6\cdot27/12=13.5$ available in closed form by the same argument.

The mechanism is not particular to the sample variance.  Any statistic that attains its maximum
at these vertices with a nonvanishing gradient inherits the same simplex geometry, and only the
parent's local density exponents set the power.  The wedge statistic
$|\bar X-\tfrac12|+s$ of the companion acceptance study is the case in point: its saturation
value $\tfrac16+\tfrac1{\sqrt3}$ is attained at the same six vertices, and for the uniform parent
the identical argument gives $\bigl(162\sqrt3-270\bigr)\delta^{3}$.  A ceiling law of this family
is therefore available for the whole class of mean--spread statistics, not only for $s^2$.

\paragraph{Singularity-isolated representation.}
$F$ carries three analytically known singular features: the origin cusp
$Y^{3/4}$~\eqref{eq:singular-smallY}; the half-power regime boundary at $Y=\tfrac14$, the
geometric bifurcation of Section~\ref{sec:R3}, present for \emph{every} parent, across which $F$
is smooth on each side but nonanalytic, entering at $(4Y-1)^{5/2}$ for the uniform parent;
and the $\delta^{3/2}$ approach to the ceiling at $Y_{\max}=\tfrac13$ established
in~\eqref{eq:ceiling-law}.  Isolating all three gives the two-piece form
\begin{equation}
  F(Y)=
  \begin{cases}
    \displaystyle Y^{3/4}\sum_{k=0}^{K} a_k\,Y^{k}, & 0\le Y< \tfrac14,\\[10pt]
    \displaystyle 1-(\tfrac13-Y)^{3/2}\sum_{k=0}^{K} b_k\,(\tfrac13-Y)^{k}, & \tfrac14\le Y\le \tfrac13 ,
  \end{cases}
  \label{eq:singular-isolated}
\end{equation}
with $b_0$ fixed by \eqref{eq:ceiling-law} rather than fitted.  The two cases are non-overlapping,
and continuity at $Y=\tfrac14$ is imposed as an exact linear constraint during fitting.
The singular prefactors are fixed by the geometry and the parent; the remaining smooth
coefficients are then determined by collocation at a few points, each a one-dimensional
quadrature, or by moment matching.  This is the singular-parent counterpart of polynomial closure: where the
polynomial case \emph{is} a finite sum, the singular case is a finite sum times an explicit
algebraic prefactor.

\paragraph{Accuracy in a handful of terms, versus Royen.}
Each piece of~\eqref{eq:singular-isolated} is a smooth function times its isolated singularity,
and the truncation error falls quickly with the number of terms.  Measured against the reference
quadrature on a validation grid \emph{disjoint} from the fitting nodes, the maximum error over that grid, which spans $[0.0055,\tfrac13]$, is $2.99\times10^{-2}$ at $2$ terms, $4.1\times10^{-3}$ at $4$,
$1.05\times10^{-3}$ at $6$, $6.5\times10^{-4}$ at $8$, $3.8\times10^{-4}$ at $10$ and
$2.4\times10^{-4}$ at $12$ (Figure~\ref{fig:arcsine-conv}).  Beyond six terms the error falls by
a roughly constant factor of $0.78$ per added term; the range tested is too short to separate
that from a steep algebraic rate near $M^{-2.1}$, and we claim only the measured behavior.
Written out, the six-term approximation, three coefficients per regime, is
\begin{equation}
  \widehat F(Y)=
  \begin{cases}
    2.5957\,Y^{3/4}-0.9254\,Y^{7/4}+2.9043\,Y^{11/4}, & 0\le Y<\tfrac14,\\[6pt]
    1-\delta^{3/2}\bigl(2.978222+4.3457\,\delta+117.03\,\delta^{2}\bigr),\ \ \delta=\tfrac13-Y,
      & \tfrac14\le Y\le\tfrac13 .
  \end{cases}
  \label{eq:singular-6term}
\end{equation}
The leading upper coefficient is the exact constant of~\eqref{eq:ceiling-law} and is not fitted;
the remaining five are obtained by constrained least squares against the reference quadrature of
Appendix~\ref{app:arcsine-ref}, subject to exact continuity at $Y=\tfrac14$.  Measured on a
validation grid disjoint from the fitting nodes, the maximum error is $1.05\times10^{-3}$ on $[0.0055,\tfrac13]$, comprising $1.05\times10^{-3}$ on the lower branch and $4.4\times10^{-4}$ on the
upper; the residual jump at $Y=\tfrac14$ is at the level of arithmetic rounding.  The lower leading coefficient $2.5957$ is fitted and is \emph{not} the small-deviation constant $C_3$
of~\eqref{eq:singular-smallY}; a three-term fit over the whole regime trades leading accuracy at
$Y\to0$ for a uniformly small error, and pinning it to $C_3$ instead degrades the sup-norm error
by roughly a factor of five.  Every coefficient here is reproduced by a script in the
reference implementation, which asserts each of them
rather than printing it.  Royen's universal series
\eqref{eq:royen-series} for the same arcsine CDF converges \emph{algebraically}.  We evaluate it
as Royen does, with the terms free of $x$ summed in closed form~\cite[eq.~(9)]{royen2008}: they
add up to one minus the ratio of the mean of $Q$ to $\sup Q$, here
$1-\tfrac{1/4}{2/3}=\tfrac58$, so that only the cosine terms are truncated.  The coefficients are
computed by deterministic quadrature from the same reference, so what is measured is the
truncation error of the series alone.  On the same validation grid it reaches
$1.35\times10^{-3}$ at $80$ terms, $6.6\times10^{-4}$ at $160$ and $1.5\times10^{-4}$ at $320$, an
observed decay of $M^{-1.73}$ fitted over $M \in [20, 320]$; we quote the rate as measured and do
not derive it.  The error is not monotone in $M$: $10^{-3}$ is first reached at $75$ terms and
holds from $109$ on.  Isolating the singularity thus replaces about a hundred
oscillatory-integral terms by about half a dozen.

\begin{figure}[t]
\centering
\includegraphics[width=0.74\linewidth]{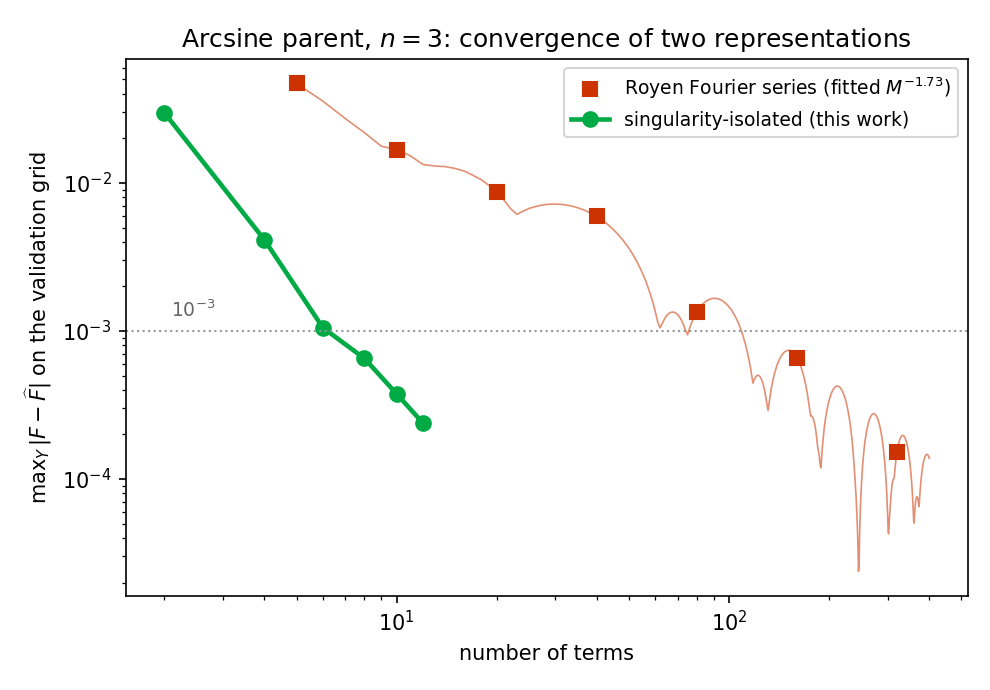}
\caption{Convergence of two representations of the arcsine $n=3$ CDF, both scored against the
deterministic quadrature reference of Appendix~\ref{app:arcsine-ref} on a grid disjoint from the
fitting nodes.  The singularity-isolated form~\eqref{eq:singular-isolated} (green) reaches
$10^{-3}$ in about six terms; Royen's Fourier series~\eqref{eq:royen-series} (red), evaluated
with its $x$-free constant in closed form, decays algebraically at an observed $M^{-1.73}$, fitted
over $M \in [20,320]$, reaching $10^{-3}$ at about a hundred terms and $1.5\times10^{-4}$ at
$320$.  The thin red line is the error at every $M$, which is not monotone; squares mark
$M=5,10,\dots,320$.  The Royen coefficients are computed by deterministic quadrature from the
same reference, so the red curve is the truncation error of the series alone, and nothing in this
figure is simulated.}
\label{fig:arcsine-conv}
\end{figure}

\paragraph{Completeness over parents.}
Together, Sections~\ref{sec:polynomial} and~\ref{sec:singular} make the $n=3$ CDF computable for
\emph{every} bounded parent: a finite elementary closed form when the density is polynomial, and
a six-term singularity-isolated approximation, accurate to about $10^{-3}$, when it is singular.  A
bounded density that is neither---the rational parent above, say---is reached by polynomial
approximation: if $\|f - g\|_\infty \le \varepsilon$ on $[0,1]$ with $g$ polynomial and
$M = \max(\|f\|_\infty, \|g\|_\infty)$, the integrands in~\eqref{eq:CDF-sextant} differ by at
most $3M^2\varepsilon$ pointwise, so $|F_f(Y) - F_g(Y)| \le 3M^2\varepsilon$ uniformly in $Y$,
and $F_g$ is elementary by Theorem~\ref{thm:poly-closure}; Weierstrass approximation therefore
extends exact computability, with a certified error bound, to every continuous bounded parent.
Taking $g$ to be the degree-$m$ Bernstein polynomial of $f$ makes the approximant a
\emph{bona fide} density---a convex combination of the Beta$(k{+}1, m{-}k{+}1)$ densities,
hence non-negative and normalized without any constraint being imposed---so $F_g$ is itself a
distribution function rather than merely an approximating integral, and
Proposition~\ref{prop:mixtures} covers it as a finite mixture.
This
is the completeness the normal-theory $\chi^2$ approximation (Section~\ref{sec:intro}) lacks, and
which Royen's series delivers only at the cost of about a hundred terms.

%=====================================================================
\section{Discussion}\label{sec:discussion}
%=====================================================================

\subsection{Why the diagonal-orthogonal coordinates}

Working throughout in this single coordinate system keeps the derivation
transparent.  Two features carry the paper:
\begin{enumerate}[topsep=2pt]
  \item \textbf{Fewer pieces.} The uniform CDF is two pieces (a single
    bifurcation at $Y = 1/4$), against the three-branch classical form.
  \item \textbf{Transparent obstructions.} The candidate obstruction is read
    directly off the sextant integrand: each parent produces a characteristic
    primitive ($\ln$, $\operatorname{Li}_2$, hyperelliptic integrals,
    $\operatorname{Ei}$, $I_0$), as cataloged---with its proved-versus-expected
    status per row---in the closure hierarchy (Section~\ref{sec:hierarchy}).
\end{enumerate}

\subsection{Beyond the variance marginal}\label{sec:beyond}

The same geometry yields more than the marginal treated here.  Conditioning
on both the mean and the spread reduces Springer's constrained integral to a
weighted arc of the deviation circle, and for the uniform parent this gives
the \emph{joint} law of $(\bar X, s)$ in closed form; the derivation and its
acceptance-value application are given in~\cite{enginecompanion}.

Whether the joint law admits the same classification established here for the
variance marginal is open.  Theorem~\ref{thm:poly-closure} closes every
polynomial parent for $F$, but its proof turns on the radial integral of the
sextant chart, and conditioning on $\bar X$ replaces that radial integral by
an arc-length measure whose weight depends on the parent along the arc.  We
have not derived the Beta cases, and we do not claim the polynomial row of
Table~\ref{tab:hierarchy} transfers.

\subsection{Extension to larger sample sizes}

The same diagonal-orthogonal reduction extends to $n \ge 4$: the variance
constraint becomes a ball in the $(n-1)$-dimensional contrast space and the
cube cross-section a polytope with $S_n$ symmetry.  The $n = 4$ case is
treated in~\cite{n4companion}.

What extends is the reduction, not the conclusion.  Theorem~\ref{thm:poly-closure}
is a statement about $n = 3$, and the closure it establishes should not be
assumed to persist as $n$ grows.  A single global frame is also not obviously
the right choice at larger $n$: coordinates adapted to slicing the polytope
recursively, one facet family at a time, may organize the cut structure better
than the one used here.

\subsection{Application to acceptance testing}\label{sec:apply}

The closed forms above have a direct use: given a within-batch density
$f$ estimated from process data, $F(Y) = \Pr(s^2 \le Y)$ is the exact
probability that the sample variance of $n$ dosage units falls below any
threshold, with no simulation and no series truncation.

Content uniformity is the motivating case.  The standard
treatments---Bergum's method~\cite{bergum1990, bergumLi2007} and the
ASTM E2709/E2810 standards~\cite{astmE2709, astmE2810}---assume a
Gaussian parent.  That assumption is convenient and physically wrong:
dosage-unit content is bounded below by zero and above by a physical
maximum, and potent low-dose tablets show positively skewed content
distributions~\cite{orr1978}, shapes that first-principles
granular-mixture models reproduce within the Beta family~\cite{rane2012}.
Beta$(a,b)$ on a bounded interval is therefore the natural parametric
class here, and Theorem~\ref{thm:poly-closure} guarantees that every
integer-parameter member has a finite elementary $s^2$ CDF\@.  The
Beta$(2,1)$, Beta$(2,2)$ and Beta$(3,2)$ examples of
\S\ref{sec:beta-examples} are not arbitrary: they are realistic
within-batch densities for moderate-dose tablets with slight to moderate
asymmetry.

Two practical remarks.  First, the parent is now observable: NIR
real-time release testing~\cite{peinado2014, karner2023} yields
individual-unit assays by the hundreds or thousands per batch, enough to
fit $f$ directly, and stochastic tablet-formation models predict it from drug-substance
bulk properties such as the API particle-size distribution~\cite{morrison2021,chu2025}.  Second, $n = 3$ sits
below the $n = 10$ of USP $\langle 905\rangle$~\cite{usp905}: the
formulas here are exactly solvable reference points, while the
general-$n$ evaluation engine of the companion
paper~\cite{enginecompanion} carries the same construction to the sample
sizes and to the acceptance-value statistic
$\mathrm{AV} = |M - \bar{X}| + k\,s$ that the standard actually uses.
Nor is $n = 3$ merely a mathematical warm-up: European concrete
conformity assessment judges \emph{initial production} on
non-overlapping groups of exactly three consecutive strength results
(the mean of three against $f_{ck} + 4$\,MPa, with an individual floor
at $f_{ck} - 4$)~\cite{en206}; that criterion is mean-based rather than
variance-based, but it makes $n = 3$ a sample size written into a
production standard in force.  The same architecture appears in
pharmaceutical work: a small-scale dissolution screen proposed for
identifying substandard and falsified medicines accepts on three units
when the mean dissolution is at least $Q + 6\%$ and the minimum at least
$Q + 2\%$~\cite{rahman2021}, the limits being calibrated by resampling
triples from batches of twenty-four.  Neither rule is a variance rule,
but both pair a mean against an individual floor at $n = 3$, so the
pass probability of either is a functional of the joint law of
$(\bar X, s)$ rather than of the variance marginal alone---the object of
\S\ref{sec:beyond} and of~\cite{enginecompanion}, for which the CDF derived
here is the $s$-margin.

The other settings named in \S\ref{sec:intro} inherit the same
machinery.  In capability approval the incumbent is normal-theory
distribution theory for $C_{pk}$~\cite{pearn1992}, extended to
non-normal parents by Clements' percentile
construction~\cite{clements1989}; an exact $s^2$ law under a bounded
parent replaces both.  In reference-scaled average
bioequivalence~\cite{haidar2008} the acceptance criterion
$(\bar{X}_T - \bar{X}_R)^2 - \theta\,s_{WR}^2 \le 0$ is again a wedge in
a (location, spread) plane, but oriented the other way---it \emph{widens}
with $s_{WR}$ where the content-uniformity wedge narrows with $s$.

For each of these the closure hierarchy is the roadmap: identify the
minimal class of the parent, read the proved status and the expected
special-function class off Table~\ref{tab:hierarchy}, and either
evaluate in closed form (the polynomial row) or choose the numerical
method that the conjectured obstruction dictates.

\subsection{Exact formulas and numerical evaluation}\label{sec:exact-vs-numeric}

That last choice deserves stating plainly, because closed forms and
certified numerics answer different questions and the results above make
the division concrete.

Three things here are not available from a numerical evaluator at any
cost.  The first is the closure hierarchy of \S\ref{sec:hierarchy}:
whether a CDF is elementary, and which special function obstructs it, is
a property of the integral rather than of its values, and no amount of
accurate evaluation decides it.  The second is the endpoint laws of
\S\ref{sec:singular}.  The origin cusp
in~\eqref{eq:singular-smallY} and the ceiling exponent
in~\eqref{eq:ceiling-law} are analytic facts about the two extreme
configurations of the sample, and once they are isolated the arcsine
parent reaches $10^{-3}$ in six terms where the same accuracy from
Royen's universal series takes of order $10^{2}$
(Figure~\ref{fig:arcsine-conv}).  Analytic structure is not a
convenience there; it is the difference between half a dozen terms and a
hundred.  The third is exact differentiation: $F'$ is available in
closed form, hence the density, and through
\eqref{eq:sd-cdf}--\eqref{eq:sd-pdf} the standard-deviation density as
well, with no differencing error propagating into quantiles, moments or
sensitivities.

Against that, numerical evaluation does what closed forms cannot, and
the gap is largest exactly where the regulations are.  Acceptance rules
in force are written at $n = 10$ and $n = 30$, not at $n = 3$, and a
parent fitted to real assay data is rarely exactly polynomial.  The
production tool at those sample sizes is the general-$n$
engine~\cite{enginecompanion}; what the closed forms supply is the
reference it is validated against, which is the role the exact column of
Table~\ref{tab:chi2-fail} plays here and the role the $n = 4$
case~\cite{n4companion} extends.  The Weierstrass argument closing
\S\ref{sec:singular} is the same division stated generally: exact
computability reaches every continuous bounded parent, but through an
approximation step whose error has to be certified rather than assumed.
The closed forms are the calibration standard, not the production
instrument, and a series of small-$n$ cases is worth having on that
ground even where a rule samples more units.  Not every rule does: the
concrete conformity standard discussed above~\cite{en206} judges groups
of exactly three.

\paragraph{What exactness does not buy.}
None of this makes the exact value the dominant term in an application.
An acceptance probability computed in practice is conditional on a
parent estimated from finite data, and propagating that estimation
uncertainty produces an interval on the probability itself that is wider
by orders of magnitude than any evaluation error considered
here~\cite{propagationcompanion}.  Exactness earns its place not by
being the largest correction available but by being the one that can be
removed \emph{completely}: once the evaluation error is gone, the
uncertainty in the parent is the only term left, and it can be reported
as what it is instead of being confounded with numerical error.  A
closed form is therefore most valuable not where it is most accurate but
where it makes the remaining uncertainty legible.

%=====================================================================
\section*{Code availability}

The closed forms of this paper are implemented in a small reference
package, to be released under the MIT license: the piecewise-elementary
CDFs, the Monte-Carlo cross-checks reported above, a notebook rendering
each closed form alongside its numerical verification, and the scripts
that regenerate every figure.  The implementation requires only NumPy; the
notebook additionally uses Matplotlib.  Every formula in this paper is
stated in full and can be implemented directly from the text; the
package accompanies the published version of this article and will be
added to a later version of this preprint.

\section*{Data availability}

No data were generated or analyzed.  Every parent distribution appearing
in this paper is an analytic law stated in the text, so the closed forms
above constitute the complete computational content of the article.

\ifblind\else
\section*{Acknowledgments}

We thank Nicole Tin, Somdatta Bhattacharya and Chiajen Lai for
extensive conversation about the pharmaceutical applications and
in-silico tools that motivated this work. We thank Mandeep Chauhan,
Mike Salem, Ash Baghai, Wei-Tin Liu, Shikhar Mohan, Adedayo Sanni, Chris Remmers, Caroline Branch, Kiara Cui,
Dana Caulder, Joanna Koziara, Victor Rucker, Jennifer Hedborn and
Henry Morrison for motivating discussions at Gilead Sciences.
\fi

\section*{Declarations}

\ifblind
The authors' affiliations and competing interests are withheld for
anonymous review.
\else
R.~Osan and R.~Yu are employees of Gilead Sciences; K.~T. Chu is an
employee of Velexi Corporation.  The authors have no other competing
interests to declare.
\fi
\medskip
\noindent\textbf{Funding.}  No external funding was received for this work.

\section*{Declaration of generative AI use}
Claude models (Anthropic) assisted with derivations, code implementation, and
manuscript drafting. All outputs were verified against independent symbolic and numerical
computations (Python, MATLAB).

%=====================================================================
% References
%=====================================================================

\clearpage
\appendix
% Appendix floats carry their own series, so the main-text figures are numbered in
% order of first citation instead of being interleaved with the atlas.
\setcounter{figure}{0}
\renewcommand{\thefigure}{A.\arabic{figure}}
\renewcommand{\theHfigure}{A.\arabic{figure}}
\section{Visual atlas of the uniform-parent derivation}\label{app:atlas}

The derivation of Section~\ref{sec:cdf-uniform} is fully specified by its formulas,
but each zone has a picture. This appendix collects the integrand-level
visualizations: the cross-section sweep, the shaving integrands of the thin-cylinder
regime, and the tip / partial-clip / saturation zones of the wide-cylinder regime
with the vertex arc-shrinking mechanism.

\begin{figure}[p]
\centering
\includegraphics[width=\linewidth]{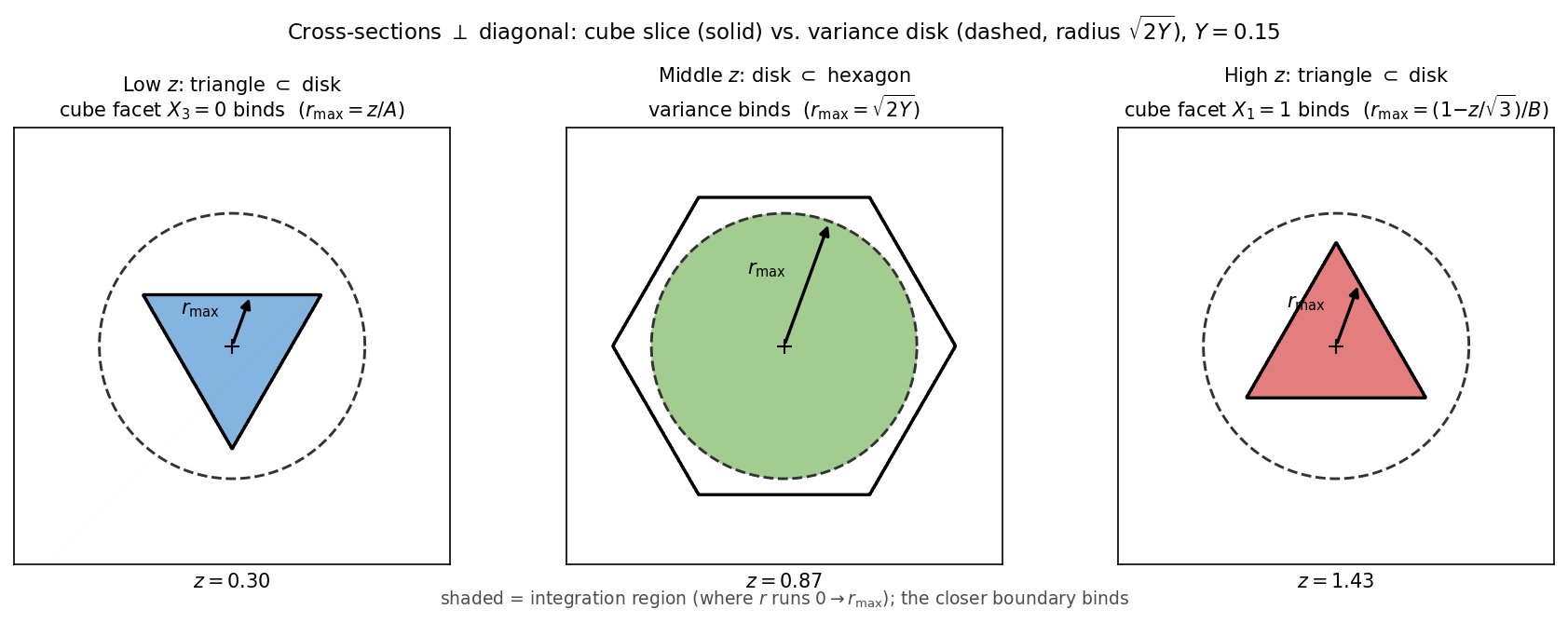}
\caption{Cross-sections perpendicular to the diagonal at three heights $z$
($Y=0.15$).  Solid: the cube slice (triangle near a corner, hexagon in the
middle); dashed: the variance disk of radius $\sqrt{2Y}$; shaded: the
integration region, bounded by whichever is closer.  \emph{Left} (small $z$):
triangle $\subset$ disk, so the lower facet $X_3=0$ binds ($r_{\max}=z/A$).
\emph{Middle}: disk $\subset$ hexagon, so the variance binds
($r_{\max}=\sqrt{2Y}$).  \emph{Right} (large $z$): upper triangle $\subset$
disk, so $X_1=1$ binds ($r_{\max}=(1-z/\sqrt3)/B$).  The bifurcation at
$Y=\tfrac14$ is exactly when the middle disk radius $\sqrt{2Y}$ equals the central hexagon's
\emph{circumradius} $1/\sqrt2$; its apothem is the smaller $\sqrt{3/8}$, reached earlier at
$Y=3/16$.}
\label{fig:xsection}
\end{figure}

\begin{figure}[p]
\centering
\includegraphics[width=\linewidth]{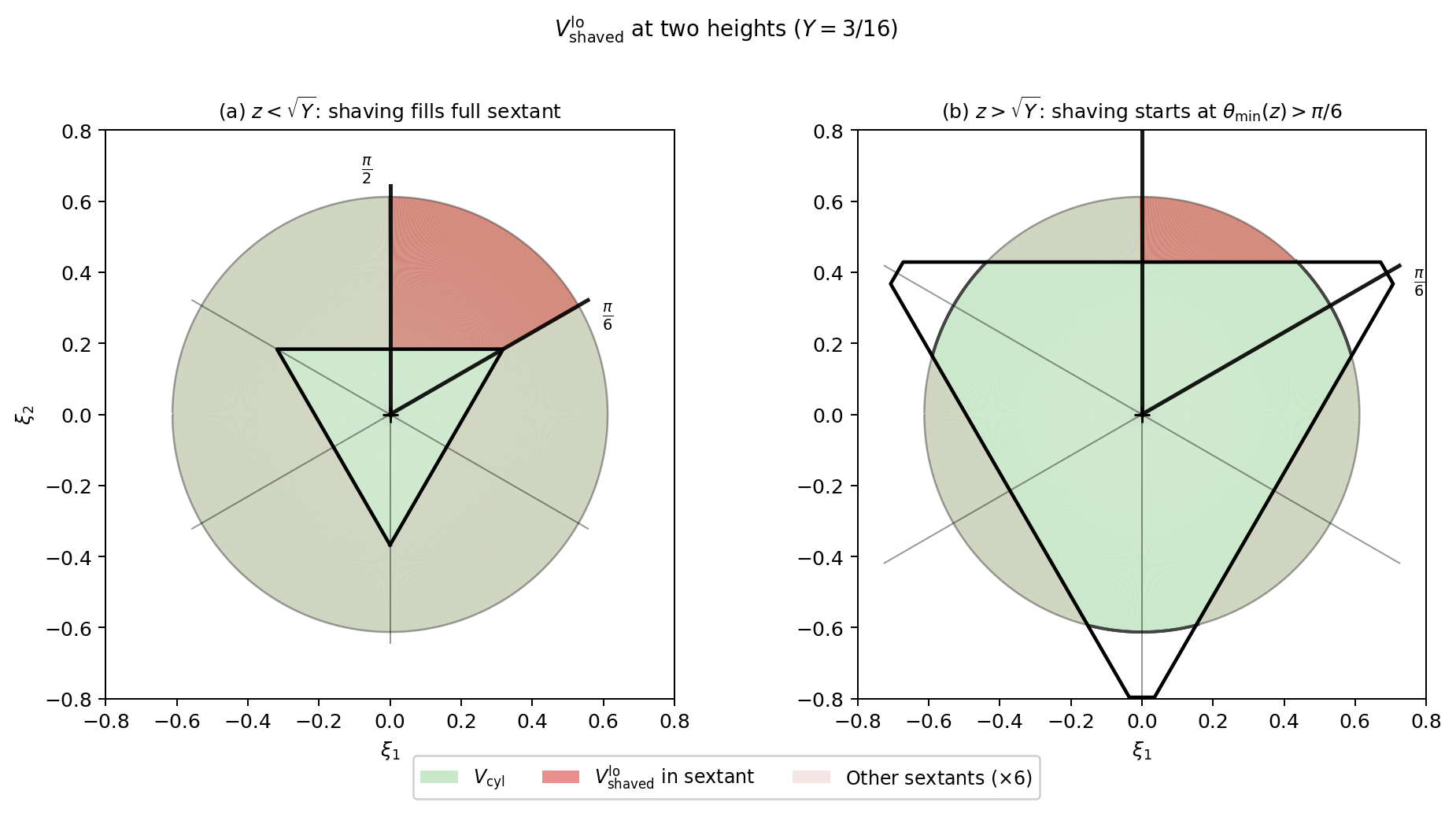}
\caption{The integrand of~\eqref{eq:Vshaved-lo-zform} at two
heights ($Y = 3/16$).  At each~$z$, the red annular area is the
shaving to be subtracted; integrating these areas over~$z$ gives
$V_{\rm shaved}^{\rm lo}$.
\textbf{(a)}~$z < \sqrt{Y}$ ($a = 0.15$):
$\theta_{\min} = \pi/6$, shaving fills the full sextant.
\textbf{(b)}~$z > \sqrt{Y}$ ($a = 0.35$):
$\theta_{\min} \approx 44^\circ > \pi/6$, shaving only near
$\theta = \pi/2$.  Pink = same shaving in the other five
sextants ($\times 6$ symmetry).
By $a \leftrightarrow 1-a$ symmetry,
$V_{\rm shaved}^{\rm hi} = V_{\rm shaved}^{\rm lo}$.}
\label{fig:shaved-lo}
\end{figure}

\begin{figure}[p]
\centering
\includegraphics[width=\linewidth]{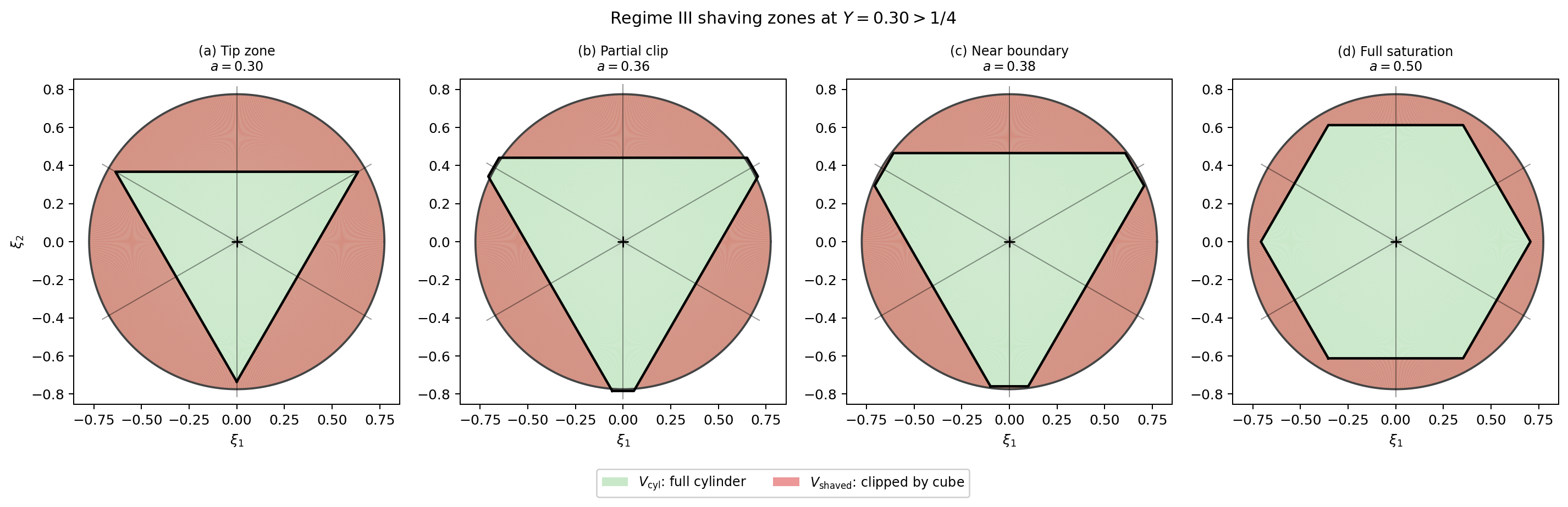}
\caption{Shaving zones at $Y = 0.30 > 1/4$, tracing the
cross-section from a cube corner toward the center.
\textbf{(a)}~Tip zone ($a = 0.30$): small triangle fully inside
the disk; shaving is the full annulus.
\textbf{(b)}~Partial clip ($a = 0.36$): the hexagon has grown past
the disk; its three protruding corners are clipped into circular
arcs.  Three short $X_i = 1$ edges (visible between the arcs) are
growing inward.
\textbf{(c)}~Near boundary ($a = 0.38 \approx a_*$): the $X_i = 1$
edges have lengthened and the circular arcs are nearly gone---the
hexagon almost fits inside the disk.
\textbf{(d)}~Full saturation ($a = 0.50$): the disk fully
circumscribes the hexagon; shaving is again a complete annulus.}
\label{fig:regime3-zones}
\end{figure}

\begin{figure}[p]
\centering
\includegraphics[width=\linewidth]{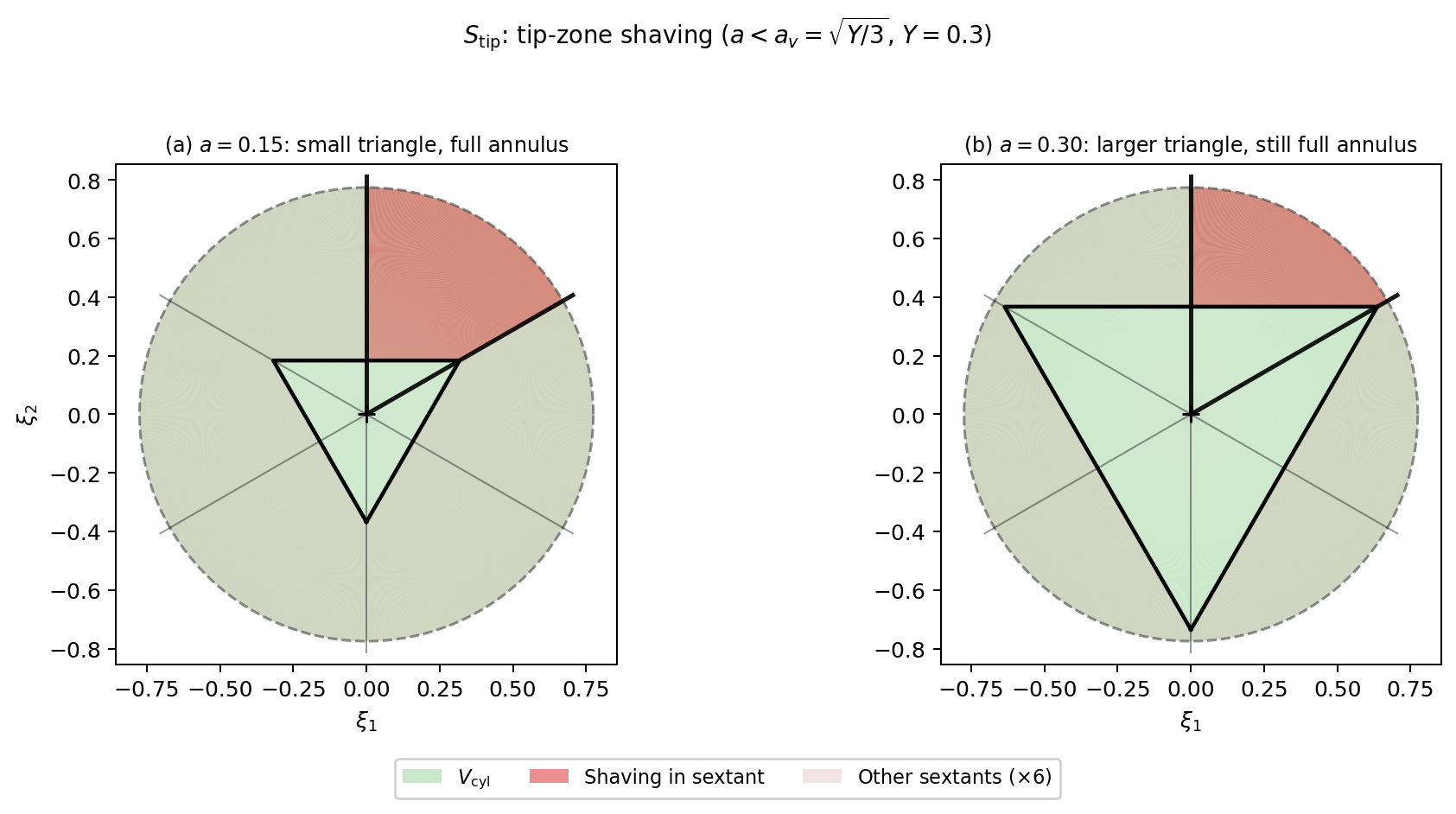}
\caption{Integrand of $S_{\rm tip}$~\eqref{eq:S-tip} at two
heights in the tip zone ($Y = 0.30$).  The triangle is fully
inside the disk; the red annulus (disk $-$ triangle) is
integrated over $z \in [0, z_v]$.}
\label{fig:zone-tip}
\end{figure}

\begin{figure}[p]
\centering
\includegraphics[width=\linewidth]{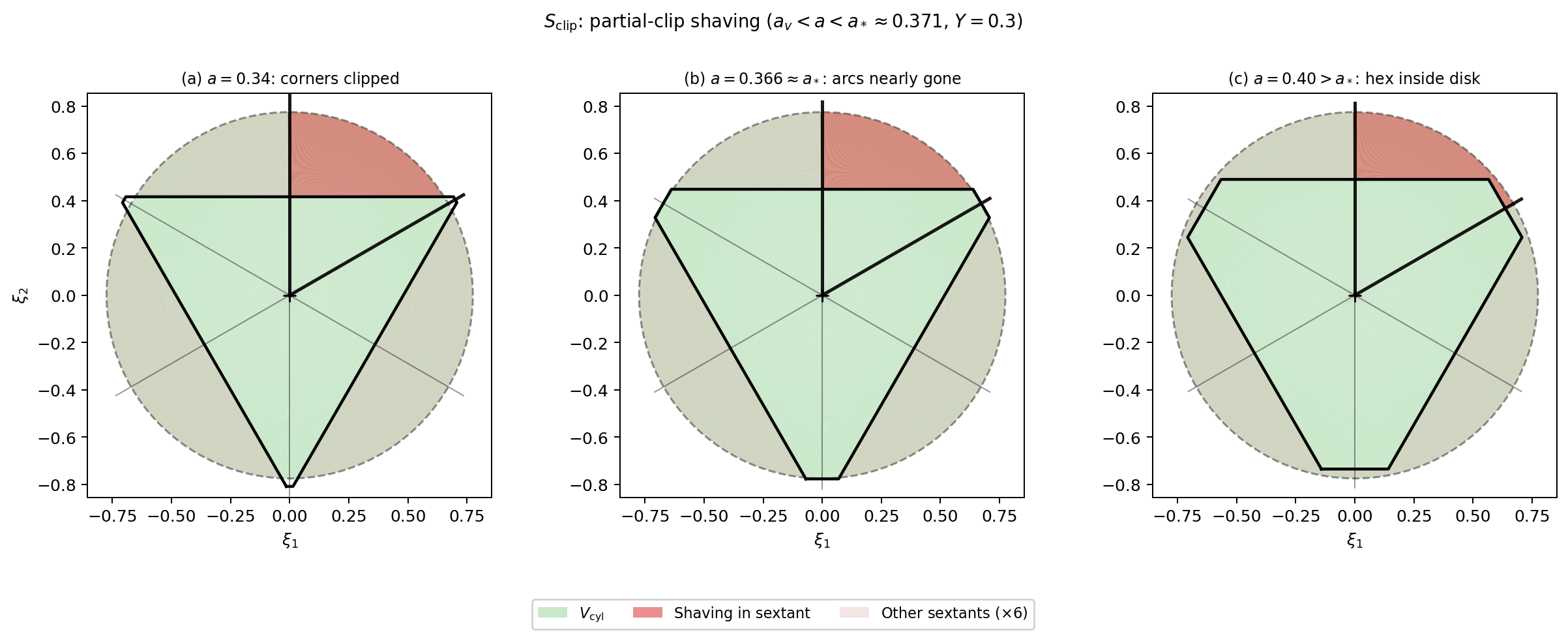}
\caption{Integrand of $S_{\rm clip}$~\eqref{eq:S-clip} at two
heights in the partial-clip zone ($Y = 0.30$).
The hexagon's corners protrude past the disk; the red segments
are the clipped portions.  As $a \to a_*$, the arcs shrink
to zero.}
\label{fig:zone-clip}
\end{figure}

\begin{figure}[p]
\centering
\includegraphics[width=\linewidth]{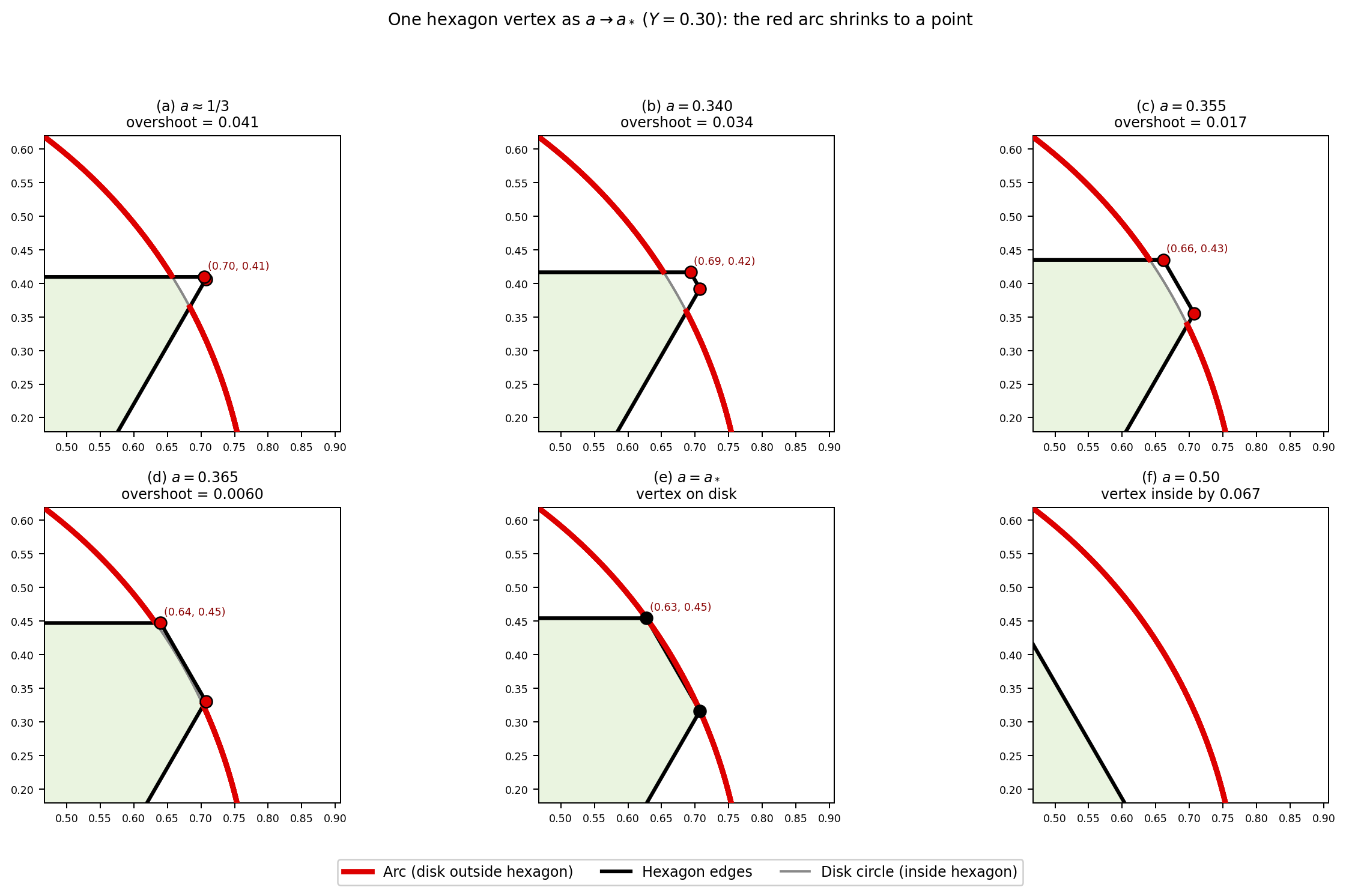}
\caption{The arc-shrinking mechanism as $a \to a_*$ ($Y = 0.30$,
$a_* \approx 0.371$).  Each hexagon has six vertices---points
where an $X_i = 0$ face edge meets an $X_j = 1$ face edge in the
cross-section plane---all at the circumradius
$\sqrt{2(1-3a+3a^2)}$ from the center.
Red dots = hexagon vertices outside the disk;
black dots = vertices on or inside the disk.  As $a$ increases,
the circumradius shrinks and the vertices approach the disk
boundary; the circular arcs at each protruding vertex contract
continuously.  In panel~(d) the vertices land exactly on the
disk circle and the arcs vanish.}
\label{fig:arc-shrinking}
\end{figure}

\begin{figure}[p]
\centering
\includegraphics[width=\linewidth]{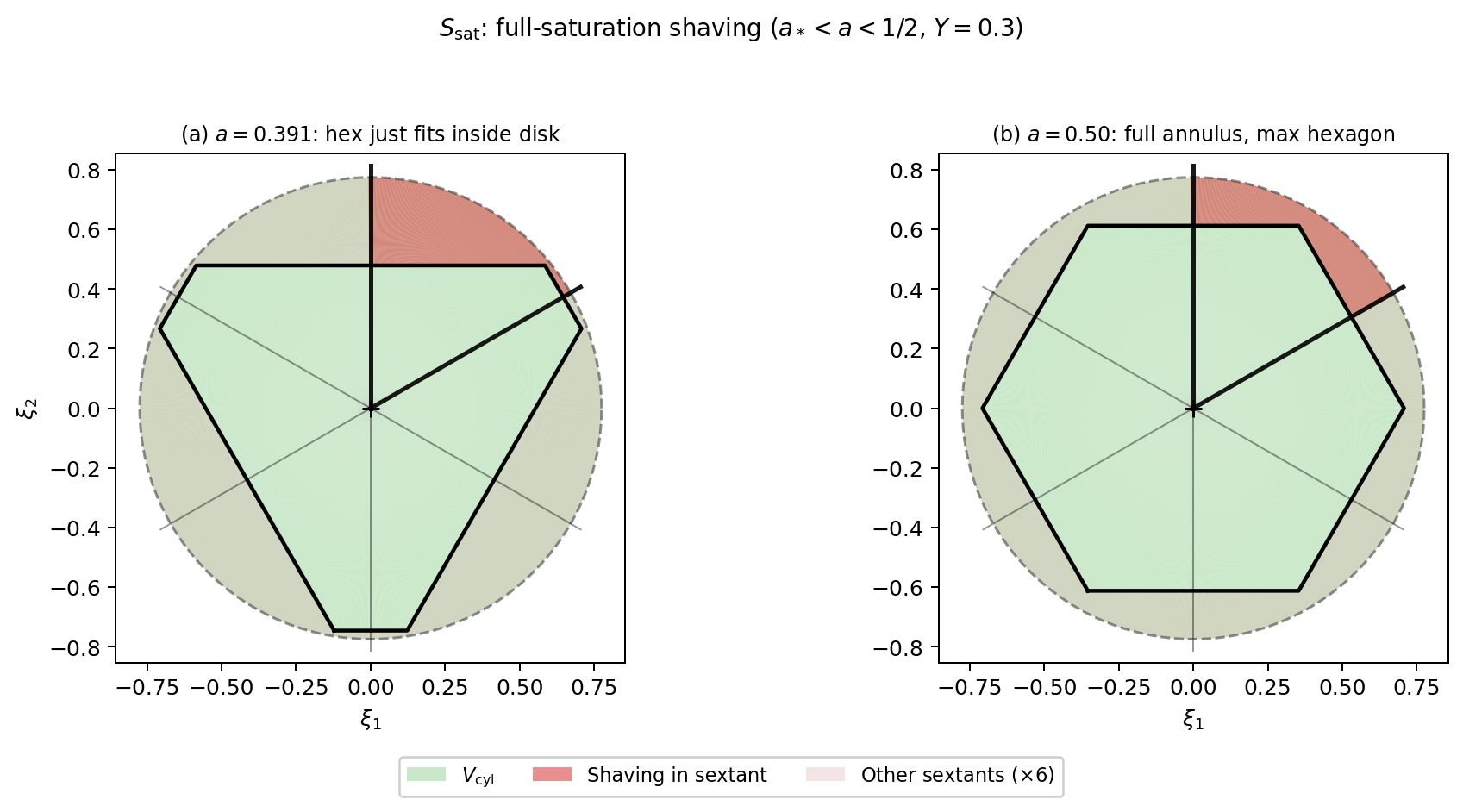}
\caption{Integrand of $S_{\rm sat}$~\eqref{eq:S-sat} at two
heights in the full-saturation zone ($Y = 0.30$).
The disk fully contains the hexagon; the red annulus
(disk $-$ hexagon) is integrated over $z \in [z_*, \sqrt{3}/2]$.}
\label{fig:zone-sat}
\end{figure}

\clearpage
\section{Proof of Theorem~\ref{thm:rational}: the rational-parent CDF is
not elementary}\label{app:rational}

\paragraph{The idea in brief.}
Elementary functions are closed under differentiation, so it suffices to show that the density
$F'$ is not elementary. Step~1 writes the density as a single integral of logarithms, and
Step~2 shows that one further derivative collapses it to sixteen terms of the form
(rational)$\,\times\,$log(rational). Integrated back up, such a sum is elementary except for
dilogarithm-type pieces $\int \ln f\, d\ln h$. Step~3 attaches to these pieces an algebraic
invariant $W$ and proves that $W$ must vanish if they are elementary. Step~4 computes $W$
exactly---every root involved is a twelfth root of unity---and finds it nonzero; Step~5 shows
that the two sectors of the computation cannot cancel each other; Step~6 carries the argument
to the wide regime $\tfrac14 < Y < \tfrac13$. The non-elementary part is therefore an
irreducible dilogarithm. Readers who want only the logic can stop here.

Throughout, \emph{elementary} is meant in the Liouville sense: the class
of functions obtained from the rational functions by finitely many
algebraic operations, exponentials and logarithms
\cite{rosenlicht1972}.  This class is closed under differentiation, so
it suffices to prove that $F'$ is not elementary on $(0, \tfrac14)$.

\paragraph{Step 1: the density as a one-variable integral.}
For $f(x) = 1/((1+x)\ln 2)$, differentiating the sextant
chart~\eqref{eq:CDF-sextant} in $Y$ leaves only the variance-disk region
(the facet regions do not depend on $Y$, and $r_{\max}$ is continuous
across region boundaries), so with $\rho = \sqrt{2Y}$,
\begin{equation}
  F'(Y) \;=\; \frac{3\sqrt{3}}{(\ln 2)^3\, Y}\; I(\rho),
  \label{eq:app-density}
\end{equation}
where the $z$-integration is elementary (the parent product depends on
$z$ only through the common factor $1 + z/\sqrt{3}$, linear in $z$) and
yields, with $s_\pm = \sin(\theta \pm \pi/6)$,
\begin{equation}
  I(\rho) = \int_{\pi/6}^{\pi/2}\!\Bigl[
    \frac{\ln 2 - \ln\bigl(1{+}\sqrt{2}\rho\, s_+\bigr)}{2\cos\theta\, s_+}
  - \frac{\ln\bigl(2{-}\sqrt{2}\rho\cos\theta\bigr)
        - \ln\bigl(1{+}\sqrt{2}\rho\, s_-\bigr)}{2\cos\theta\, s_-}
  + \frac{\ln\bigl(2{-}\sqrt{2}\rho\, s_+\bigr)}{2\, s_+ s_-}
  \Bigr]\, d\theta .
  \label{eq:app-I}
\end{equation}
Since $I(\rho) = \tfrac{(\ln 2)^3}{3\sqrt{3}}\cdot\tfrac{\rho^2}{2}\,
F'(\rho^2/2)$ and $\rho \mapsto \rho^2/2$ is algebraic, $F'$ is
elementary if and only if $I$ is.  We verified~\eqref{eq:app-density}
and \eqref{eq:app-I} against two-dimensional quadrature of the density
to $10^{-13}$ over the regime.

\paragraph{Step 2: the derivative of $I$ collapses.}
Differentiating under the integral and substituting
$t = \tan(\theta/2)$, both convergent groupings simplify drastically
(the endpoint-singular factors cancel algebraically), leaving, with
$\kappa = \sqrt{2}\rho \in (0,1)$,
\begin{equation}
  \sqrt{2}\, I'(\rho)
  = 8\kappa\!\int_{2-\sqrt{3}}^{1}\!\frac{(1+t^2)\,dt}{Q_+ Q_-}
  \;-\; 4\kappa\!\int_{2-\sqrt{3}}^{1}\!\frac{(1+t^2)\,dt}{P_1 P_2},
  \label{eq:app-Iprime}
\end{equation}
\[
  Q_\pm = (2 \mp \kappa)t^2 + 2\sqrt{3}\kappa t + (2 \pm \kappa),
  \qquad
  P_1 = (2{+}\kappa)t^2 + (2{-}\kappa), \qquad
  P_2 = (4{+}\kappa)t^2 - 2\sqrt{3}\kappa t + (4{-}\kappa).
\]
The roots of $Q_\pm$ involve $w = \sqrt{1-\kappa^2}$ and those of
$P_1, P_2$ involve $w_2 = \sqrt{4-\kappa^2}$.  The conic
substitutions $\kappa = 2u/(1+u^2)$ (so $w = (1-u^2)/(1+u^2)$) and
$\kappa = 4v/(1+v^2)$ (so $w_2 = 2(1-v^2)/(1+v^2)$) rationalize the two
sectors separately: evaluating~\eqref{eq:app-Iprime} by residues gives
\begin{equation}
  dI \;=\; \sum_{j=1}^{16} R_j\,\ln h_j \; d p_j,
  \label{eq:app-dI}
\end{equation}
a sum of sixteen terms in which $p_j \in \{u, v\}$ and
$R_j, h_j \in \mathbb{Q}(i, \sqrt{3})(p_j)$ are explicit rational
functions (each root $a$ of each quadratic contributes
$\ln(1-a) - \ln(2{-}\sqrt{3}{-}a)$ with a rational coefficient).  Every
step so far is elementary calculus; we verified~\eqref{eq:app-Iprime}
and~\eqref{eq:app-dI} numerically to $10^{-16}$.

\paragraph{Step 3: the obstruction invariant.}
Integrating each term of~\eqref{eq:app-dI} by parts,
$\int R \ln h\,dp = A \ln h - \int A\, d\ln h$ with $A = \int R\,dp$; on
a rational function field $A$ is again rational plus a
\emph{constant-coefficient} combination $\sum_l d_l \ln f_l$ (the $d_l$
are residues of $R$), and every piece of the by-parts expansion is
elementary except the pairings $\int \ln f\, d\ln h$.  The obstruction
is measured by the antisymmetric invariant
\begin{equation}
  W \;=\; \sum_{\text{terms}} \sum_l d_l\, \operatorname{div}(f_l)
  \wedge \operatorname{div}(h) \;\in\;
  \overline{\mathbb{Q}} \otimes \Lambda^2_{\mathbb{Q}}
  \bigl(\mathrm{Div}\bigr).
  \label{eq:app-W}
\end{equation}

\begin{lemma}\label{lem:liouville-dilog}
Let $K$ be the function field of a smooth curve over
$\overline{\mathbb{Q}}$, let $f_l, h_l \in K^\times$ and
$d_l \in \overline{\mathbb{Q}}$, and suppose
$S = \sum_l d_l \int \ln f_l \, d\ln h_l$ coincides on an interval with
an elementary function.  Then
$\sum_l d_l\, f_l \wedge h_l = 0$ in
$\overline{\mathbb{Q}} \otimes
\Lambda^2_{\mathbb{Q}}(K^\times\!/\overline{\mathbb{Q}}^\times)$.
\end{lemma}

\begin{proof}
Choose a $\mathbb{Q}$-basis $g_1, \dots, g_N$ of the subgroup of
$K^\times\!/\overline{\mathbb{Q}}^\times \otimes \mathbb{Q}$ generated
by the $f_l, h_l$; bilinearity reduces the claim to showing that if
$S = \sum_{a,b} \lambda_{ab} \int \ln g_a\, d\ln g_b$ is elementary
then $\lambda_{ab} = \lambda_{ba}$ for all $a \neq b$.  The logarithms
$\ell_a = \ln g_a$ are algebraically independent over $K$ by
Ostrowski's theorem, so
$S' = \sum \lambda_{ab}\, \ell_a\, g_b'/g_b$ lies in the purely
transcendental extension $E = K(\ell_1, \dots, \ell_N)$
\cite{ostrowski1946} and has degree
one as a polynomial in the $\ell$'s.  By the strong Liouville theorem
\cite{rosenlicht1968, rosenlicht1972}, $S = w_0 + \sum_\nu \mu_\nu \ln
w_\nu$ with $w_0 \in E$, $\mu_\nu$ constants and $w_\nu \in E$.

We claim $w_0$ is a polynomial of degree at most two in the $\ell$'s,
say $w_0 = \sum_{a \le b} \alpha_{ab} \ell_a \ell_b + \sum_a \beta_a
\ell_a + \beta_0$ with coefficients in $K$, and that the $w_\nu$ may be
taken in $K^\times$.  All three assertions rest on the following
residue lemma, and through it on the hypothesis that the $g_a$ are
multiplicatively independent modulo constants.

\begin{quote}
\emph{Residue lemma.} If $e \in K$ and $r_a \in \overline{\mathbb{Q}}$
satisfy $e' = \sum_a r_a\, g_a'/g_a$, then all $r_a = 0$ and
$e \in \overline{\mathbb{Q}}$.  Indeed $\mathrm{d}e = \sum_a r_a\,
\mathrm{d}g_a/g_a$ as rational $1$-forms on the curve; an exact form
has vanishing residue at every closed point, while
$\operatorname{Res}_P(\mathrm{d}g/g) = \operatorname{ord}_P(g)$, so
$\sum_a r_a \operatorname{ord}_P(g_a) = 0$ for all $P$.  The rational
nullspace of the integer matrix $(\operatorname{ord}_P(g_a))$ spans the
$\overline{\mathbb{Q}}$-nullspace, so we may take $r_a \in \mathbb{Q}$;
clearing denominators, $\prod_a g_a^{M r_a}$ has divisor zero and hence
is constant, so all $M r_a = 0$ by multiplicative independence.
\end{quote}

The lemma yields, for each $i$, that no element of
$K(\ell_1, \dots, \ell_{i-1})$ has derivative in
$\overline{\mathbb{Q}}^\times\, g_i'/g_i$.  A simultaneous induction up
the tower $K \subset K(\ell_1) \subset \cdots \subset E$ gives from this
that each $\ell_i$ is transcendental over the previous stage (so the
$\ell_a$ are algebraically independent, which is Ostrowski's theorem
and need not be assumed), that every monic irreducible of the
polynomial ring is normal, and that $\operatorname{Con} E =
\overline{\mathbb{Q}}$---so the Liouville theorem applies with no
constant extension.  The three assertions follow: if $w_0$ had a
nontrivial denominator, a monic irreducible factor $p$ of it not
dividing the numerator would give $v_p(w_0') = -(n{+}1) \le -2$ by
normality, whereas $v_p(S') \ge 0$ and $v_p(w_\nu'/w_\nu) \ge -1$ force
$v_p(w_0') \ge -1$; comparing polar parts at each such $p$, after
replacing the $\mu_\nu$ by a $\mathbb{Q}$-basis, shows every $w_\nu$ has
trivial divisor outside $K$, hence lies in $K^\times$; and if $w_0$ had
total degree $D \ge 3$, the top-degree coefficients would be constants
and the residue lemma applied to the coefficients of weight $D-1$ would
kill all of weight $D$, a contradiction.  The case $|\alpha| = 2$ of the
same coefficient comparison gives $\alpha_{ab}' = 0$.

\begin{quote}
\emph{Two hypotheses worth isolating.}  Algebraic independence of the
$\ell_a$ alone would \emph{not} suffice: for $g_1 = x$, $g_2 = 2x$ over
$K = \overline{\mathbb{Q}}(x)$ the logarithms $\ln x, \ln 2x$ are
algebraically independent over $K$, yet $E$ acquires the new constant
$\ln 2$, and new constants are exactly what would break the degree and
$K$-rationality claims.  It is the $\mathbb{Q}$-basis choice---i.e.\
multiplicative independence modulo constants---that excludes this, and
it is therefore a hypothesis of the lemma rather than a convenience.
Separately, the passage from ``$S$ agrees on an interval with an
elementary function'' to the differential-algebra hypothesis is the
standard one: germs of meromorphic functions on the interval form a
differential field extending $E$ with the same constants on the
subfield generated, and an elementary function generates an elementary
extension of $E$ inside it.
\end{quote}

Differentiating and matching the coefficient of each $\ell_a$
gives
\[
  \sum_b \lambda_{ab}\, \frac{g_b'}{g_b}
  \;=\; \sum_b \alpha_{ab}\, \frac{g_b'}{g_b} \;+\; \beta_a',
  \qquad \alpha_{ab}' = 0 .
\]
The differential $\beta_a'\,dp$ is exact, so all its residues vanish;
taking residues along the divisor of $g_b$ (whose logarithmic
derivative has residues equal to the orders of $g_b$) yields
$\lambda_{ab} = \alpha_{ab}$ for $a \neq b$.  Since $\alpha$ enters
$w_0$ through the symmetric monomials $\ell_a\ell_b$, the matrix
$(\alpha_{ab})_{a\neq b}$ is symmetric; hence
$\lambda_{ab} = \lambda_{ba}$.
\end{proof}

\paragraph{Step 4: exact computation.}
The invariant~\eqref{eq:app-W} splits into the two sectors.  Computing
all residues and divisors exactly in the algebraic number field
$\mathbb{Q}(i, \sqrt{3})$---every polynomial that occurs splits into
linear factors there, with all roots twelfth roots of unity---gives,
with $\zeta = e^{i\pi/6}$ and $e_{\zeta^m}$ the place $u = \zeta^m$,
\begin{equation}
  W_A = \frac{1}{\sqrt{3}}\Bigl(
    e_{\zeta}\wedge e_{\zeta^{9}}
  + e_{\zeta^{7}}\wedge e_{\zeta^{11}}
  + e_{\zeta^{9}}\wedge e_{\zeta^{5}}
  - e_{\zeta}\wedge e_{\zeta^{5}}
  - e_{\zeta^{7}}\wedge e_{\zeta^{3}}
  - e_{\zeta^{3}}\wedge e_{\zeta^{11}}
  \Bigr) \;\neq\; 0,
  \label{eq:app-WA}
\end{equation}
and $W_B$ is the identical expression in the places $v = \zeta^m$.  An
independent $40$-digit floating-point computation of all residues and
places agrees entry for entry.

\paragraph{Step 5: no cross-sector cancellation.}
The two sectors live on different curves over
$\overline{\mathbb{Q}}(\kappa)$, joined in the compositum
$L = \overline{\mathbb{Q}}(\kappa, w, w_2)$.  Under
$\kappa = 2u/(1+u^2)$ the places $u = \zeta^{\pm 1}$ lie over $\kappa =
2/\sqrt{3}$, the places $u = \zeta^{5}, \zeta^{7}$ over $\kappa =
-2/\sqrt{3}$, and $u = \zeta^{3}, \zeta^{9}$ over $\kappa = \infty$;
every basis element of~\eqref{eq:app-WA} has at least one slot at a
place over $\kappa = \pm 2/\sqrt{3}$.  Sector-$B$ functions have zeros
and poles only over $\kappa \in \{\pm 4/\sqrt{3}, \infty, 0\}$, so
their divisors are supported away from every fibre $\kappa = \pm
2/\sqrt{3}$, and pulling both sectors back to $L$ (divisor maps are
injective on $K^\times\!/\overline{\mathbb{Q}}^\times$, and fibres over
distinct $\kappa$-values are disjoint) shows that no
$\Lambda^2$-combination from sector $B$ can meet the sector-$A$ basis
elements.  Hence $W \neq 0$ in
$\overline{\mathbb{Q}} \otimes \Lambda^2_{\mathbb{Q}}
(L^\times\!/\overline{\mathbb{Q}}^\times)$.

This proves the theorem for the low branch: if $F|_{(0,1/4)}$ were
elementary, then $F'$,
hence $I$, would be elementary; by Steps 2--3 the dilogarithmic part of
$I$ would be elementary, so by Lemma~\ref{lem:liouville-dilog} the
invariant $W$ would vanish---contradicting Steps 4--5.

\paragraph{Step 6: the wide regime.}
For $Y \in (\tfrac14, \tfrac13)$, i.e.\ $\kappa > 1$, the variance band
pinches off where $\sin(\theta + \pi/6) = 1/\kappa$, and
\eqref{eq:app-I} holds with the $\theta$-range restricted to the
complement of the pinch window.  In the variable $t = \tan(\theta/2)$
the pinch condition is the quadratic
$(2{+}\kappa)t^2 - 2\sqrt{3}\kappa t + (2{-}\kappa) = 0$, whose roots
\[
  \tau_{1,2} \;=\; \frac{\sqrt{3}\kappa \mp 2\sqrt{\kappa^2 - 1}}
  {2 + \kappa}
  \;=\; \frac{\sqrt{3}u \mp i(1 - u^2)}{u^2 + u + 1}
\]
are the mirrors $t \mapsto -t$ of the roots of $Q_-$ and remain
rational on the sector-$A$ curve, since
$\sqrt{\kappa^2 - 1} = i\sqrt{1 - \kappa^2}$ generates no new
extension.  The integrand of~\eqref{eq:app-I} vanishes at the pinch, so
no boundary terms arise, and~\eqref{eq:app-Iprime} holds with the
$t$-range $[2{-}\sqrt{3}, \tau_1] \cup [\tau_2, 1]$; we verified the
wide-regime density to $10^{-11}$ and the restricted collapse against a
finite difference.  The wide-regime
representation~\eqref{eq:app-dI} therefore consists of the low-regime
terms plus, for each root $a$, the extra logarithms
$\ln(\tau_1 - a) - \ln(\tau_2 - a)$.  These extra terms cannot disturb
the entries $e_{\zeta^{7}}\wedge e_{\zeta^{11}}$ and
$e_{\zeta}\wedge e_{\zeta^{5}}$ of~\eqref{eq:app-WA}: on sector $A$,
all eight numerators of $\tau_i - a$ factor over
$\mathbb{Q}(i, \sqrt{3})$ with zeros only in
$\{0, \pm i, \zeta^{\pm 2}, \zeta^{\pm 4}\}$ and poles at
$\zeta^{\pm 2}, \zeta^{\pm 4}$---never at
$\zeta^{\pm 1}, \zeta^{\pm 5}$---while on sector $B$ the divisors of
$\tau_i - b$ have no support over the fibres $\kappa = \pm 2/\sqrt{3}$,
because there the $\tau_i$ are real and the roots of $P_1 P_2$ are
strictly complex.  Both slots of the two entries above lie over
$\kappa = \pm 2/\sqrt{3}$, so they survive verbatim, the wide-regime
invariant is nonzero, and $F|_{(1/4,1/3)}$ is not elementary either.
\hfill$\square$

\paragraph{Why the two sectors carry the same pattern.}
The identity $P_2(t; \kappa) = 2\,Q_-(-t;\, \kappa/2)$, together with
$P_1(t;\kappa) = 2\bigl[(1{+}t^2) - \tfrac{\kappa}{2}(1{-}t^2)\bigr]$,
exhibits sector $B$ as the sector-$A$ template---a pair of shifted
sines at angular gap $\pi/3$---evaluated at half the modulus,
$\kappa \mapsto \kappa/2$, and reflected.  The invariant of the
template depends only on its discrete data (the $\pi/3$ gap, which
generates the twelfth-roots-of-unity geometry), not on the modulus
parametrization, which is why $W_A$ and $W_B$ exhibit the identical
cyclotomic pattern.

\medskip
\noindent
Baddoura's theory of integration in finite terms with dilogarithms
\cite{baddoura2006} provides the general context for
Lemma~\ref{lem:liouville-dilog}; the proof given here is
self-contained.  Scripts reproducing every step of the computation
(the reduction, the sixteen-term representation, the exact residues
over $\mathbb{Q}(i,\sqrt{3})$, and the floating-point cross-check) are
part of the reference implementation.

%=====================================================================
\section{Reference computation and endpoint laws for the arcsine
parent}\label{app:arcsine-ref}
%=====================================================================

Every arcsine number quoted in Section~\ref{sec:singular} rests on a reference CDF computed
without simulation, and on two endpoint laws derived in closed form.  This appendix states the
reference method, its validation, and the provenance of each constant.

\paragraph{Reference CDF by two-dimensional quadrature.}
For fixed $X_1=x_1$ and $X_2=x_2$ the map $x_3\mapsto s^2$ is a convex quadratic, so
$\{x_3: s^2\le Y\}$ is an interval with endpoints
\begin{equation}
  \tfrac12\Bigl[(x_1+x_2)\mp\sqrt{12Y-3(x_1-x_2)^2}\,\Bigr],
  \label{eq:x3-interval}
\end{equation}
empty when the radicand is negative.  Hence
$F(Y)=\mathbb E\bigl[G(\mathrm{hi})-G(\mathrm{lo})\bigr]$ with $G$ the parent CDF and the
endpoints clipped to $[0,1]$.  For the arcsine, substituting $x=\sin^2(\pi t/2)$ maps the
expectation to a flat integral over $t\in[0,1]^2$ and removes both endpoint singularities
exactly, leaving a bounded integrand suitable for tensor Gauss--Legendre quadrature.  Doubling
the node count from $1500$ to $3000$ per axis changes $F$ in the eighth decimal, and the values
used here are quoted to that precision.  This construction uses no random numbers and is
independent of the Monte Carlo checks below.

\paragraph{The ceiling law.}
Equation~\eqref{eq:ceiling-law} follows from the exact corner
expansion~\eqref{eq:corner-linear} together with the Dirichlet integral
\begin{equation}
  \int_{\{p+q+2r<1,\;p,q,r>0\}}(pqr)^{-1/2}\,dp\,dq\,dr
  \;=\;\frac{1}{\sqrt2}\cdot\frac{\Gamma(\tfrac12)^3}{\Gamma(\tfrac52)}
  \;=\;\frac{4\pi}{3\sqrt2}\;=\;2.9619219588\ldots,
  \label{eq:dirichlet}
\end{equation}
the factor $1/\sqrt2$ arising from the substitution that turns $2r$ into a simplex coordinate.
Scaling by $(3\delta)^{3/2}$, dividing by $\pi^3$ and summing the six $2$--$1$ vertices
gives the coefficient $24\sqrt{3/2}/\pi^2$.  Symbolic and numerical evaluations
of~\eqref{eq:dirichlet} agree to $5\times10^{-14}$.

Two independent checks confirm the law.  A simulation of $2\times10^{8}$ samples gives a ratio of
the empirical tail to~\eqref{eq:ceiling-law} of $1.0960$, $1.0299$, $1.0065$, $1.0033$ and
$0.9993$ at $\delta=3\times10^{-2}$, $10^{-2}$, $3\times10^{-3}$, $10^{-3}$ and
$3\times10^{-4}$, with a log--log slope of $1.5031$ over the three smallest values.  The
quadrature reference reproduces the same ratios to within $0.06\%$.  For the uniform parent the
same argument predicts exponent $3$ with coefficient $13.5$, which simulation confirms.

\paragraph{The origin constant.}
$C_3$ in~\eqref{eq:singular-smallY} is obtained from the quadrature reference by Richardson
extrapolation of $F(Y)/Y^{3/4}$ in powers of $Y^{1/4}$ over the five nodes
$Y\in[10^{-5},10^{-3}]$; quadratic and cubic extrapolants give $2.6958$ and $2.6935$, and we
quote $C_3=2.694\pm0.002$.  Unlike the ceiling coefficient, no closed form for $C_3$ is known to
us: the corner integral that produces it diverges along the diagonal direction and is cut off by
the support boundary, so the constant is not given by a single local computation.

\paragraph{Fitting protocol for~\eqref{eq:singular-6term}.}
Fitting nodes are $40$ points equispaced on $[0.004,\tfrac14)$ and $28$ on
$[\tfrac14,\tfrac13)$; the validation nodes are $37$ and $25$ points on offset grids sharing no
point with the fitting set.  The objective is unweighted least squares on the reference values,
subject to two exact linear constraints: $b_0=24\sqrt{3/2}/\pi^2$, and continuity at
$Y=\tfrac14$.  The constrained problem is solved in one linear system through its
Karush--Kuhn--Tucker conditions, so the constraints hold to arithmetic precision rather than
approximately.  All errors reported in Section~\ref{sec:singular} are maximum absolute errors on
the \emph{validation} nodes.

\end{document}